\documentclass[twoside,onecolumn]{article}

\usepackage[english]{babel} % Language hyphenation and typographical rules

\usepackage[hmarginratio=1:1,top=32mm,columnsep=20pt]{geometry} % Document margins
\usepackage[hang, small,labelfont=bf,up,textfont=it,up]{caption} % Custom captions under/above floats in tables or figures
\usepackage{booktabs} % Horizontal rules in tables

\usepackage{geometry} % to change the page dimensions
\usepackage{graphicx} % support the \includegraphics command and options

\usepackage{booktabs} % for much better looking tables
\usepackage{array} % for better arrays (eg matrices) in maths
\usepackage{paralist} % very flexible & customisable lists (eg. enumerate/itemize, etc.)
\usepackage{verbatim} % adds environment for commenting out blocks of text & for better verbatim
\usepackage{subfig} % make it possible to include more than one captioned figure/table in a single float
\usepackage{amsmath}	% split equations, equation numbering
\allowdisplaybreaks[1]	% page break, but avoid when not necessary
\numberwithin{equation}{section}	% numbering equations with section
\usepackage{amsfonts}	% math symbols
\usepackage{amsthm}	% theorems, definitions, examples, ...
\usepackage{cases}
\usepackage{bbm}		% indicator function
\usepackage{braket}	% sets, brakets, ...
\usepackage{mathtools}	% norm, ...
\usepackage{fancyhdr} % This should be set AFTER setting up the page geometry
\theoremstyle{plain}
\newtheorem{example}{Example}[section]
\theoremstyle{plain}
\newtheorem{remark}{Remark}[section]
\theoremstyle{plain}
\newtheorem{prop}{Proposition}[section]
\theoremstyle{plain}
\newtheorem{lemma}{Lemma}[section]
\theoremstyle{plain}
\newtheorem{theorem}{Theorem}[section]
\theoremstyle{plain}
\newtheorem{corollary}{Corollary}[section]
\theoremstyle{plain}
\newtheorem{definition}{Definition}[section]
\theoremstyle{assumption}
\newtheorem{assumption}{Assumption}[section]
\theoremstyle{plain}

\DeclarePairedDelimiter{\abs}{\lvert}{\rvert}

\DeclareMathOperator*{\esssup}{ess\,sup}
\DeclareMathOperator*{\essinf}{ess\,inf}

\usepackage{sectsty}
\allsectionsfont{\sffamily\mdseries\upshape} % (See the fntguide.pdf for font help)
\usepackage{abstract} % Allows abstract customization
\usepackage{titlesec} % Allows customization of titles
\titleformat{\section}[block]{\large\bfseries\centering}{\thesection.}{1em}{} % Change the look of the section titles
\titleformat{\subsection}[block]{\large\bfseries}{\thesubsection.}{1em}{} % Change the look of the section titles

\usepackage[nottoc,notlof,notlot]{tocbibind} % Put the bibliography in the ToC
\usepackage[titles,subfigure]{tocloft} % Alter the style of the Table of Contents

\usepackage[a-1b]{pdfx} % per generare un file conforme alla norma PDF/A-1b
\usepackage{titling} % Customizing the title section
\usepackage{hyperref}

\hypersetup{%
    pdfpagemode={UseOutlines},
    bookmarksopen,
    pdfstartview={FitH},
    colorlinks,
    linkcolor={black},
    citecolor={black},
    urlcolor={black}
  }

\pretitle{\begin{center}\Huge\bfseries} % Article title formatting
\posttitle{\end{center}} % Article title closing formatting
\title{A BSDE-based approach for the optimal reinsurance problem under partial information} % Article title
\author{%
\textsc{Brachetta M.}\thanks{Department of Economics, University of Chieti-Pescara, Viale Pindaro, 42 - 65127 Pescara, Italy.} \thanks{Corresponding author.} \\[1ex]
\normalsize \href{mailto:matteo.brachetta@unich.it}{matteo.brachetta@unich.it}
\and
\textsc{Ceci, C.}\footnotemark[1] \\[1ex]
\normalsize \href{mailto:c.ceci@unich.it}{c.ceci@unich.it}
}
\date{} % Leave empty to omit a date

\providecommand{\keywords}[1]{\textbf{\textit{Keywords:}} #1}
\providecommand{\jelcodes}[1]{\textbf{\textit{JEL Classification codes:}} #1}
\providecommand{\msccodes}[1]{\textbf{\textit{MSC Classification codes:}} #1}

\begin{document}
\maketitle

\begin{abstract}
\noindent 
We investigate the optimal reinsurance problem under the criterion of maximizing the expected utility of terminal wealth when the insurance company has restricted information on the loss process. We propose a risk model with claim arrival intensity and claim sizes distribution affected by an unobservable environmental stochastic factor. By filtering techniques (with marked point process observations), we reduce the original problem to an equivalent stochastic control problem under full information. Since the classical Hamilton-Jacobi-Bellman approach does not apply, due to the infinite dimensionality of the filter, we choose an alternative approach based on Backward Stochastic Differential Equations (BSDEs). Precisely, we characterize the value process and the optimal reinsurance strategy in terms of the unique solution to a BSDE driven by a marked point process.
\end{abstract}

\noindent\keywords{Optimal reinsurance, partial information, stochastic control, backward stochastic differential equations.}\\
\noindent\jelcodes{G220, C610.}\\
%% G220 Insurance; Insurance Companies; Actuarial Studies
%% C610 Optimization Techniques; Programming Models; Dynamic Analysis
%% G110 Portfolio Choice; Investment Decisions
\noindent\msccodes{93E20, 91B30, 60G35, 60G57, 60J75.}\\
%% 93E20 Optimal stochastic control, 91B30 Risk theory, insurance, 91G80 Financial applications of other theories (stochastic control, calculus of variations, PDE, SPDE, dynamical systems), 60G35 Signal detection and filtering (nella 

% PROBABILITY THEORY AND STOCHASTIC PROCESSES), 60G57 Random measures, 60J75 Jump processes, 93E11 Filtering (nella sezione SYSTEMS THEORY; CONTROL)
%\noindent \textit{Declarations of interest: none.}

%--------------------------------------------------------------------------------
%	INTRODUCTION
%--------------------------------------------------------------------------------

\section{Introduction}

The aim of this paper is to investigate the optimal reinsurance problem when the insurer has only limited information at disposal. Insurance business requires very effective tools to manage risks and reinsurance arrangements are considered incisive to this end. From the operational viewpoint, a risk-sharing agreement helps the insurer reducing unexpected losses, stabilizing operating results, increasing business capacity and so on. The existing literature mostly concerns classical reinsurance contracts such as the proportional and the excess-of-loss, which were widely investigated under a variety of optimization criteria (see \cite{irgens_paulsen:optcontrol}, \cite{liuma:optreins}, \cite{BC:IME2019} and references therein). All these papers can be gathered in two main groups, depending on the underlying risk model: some authors describe the insurer's loss process as a diffusion model (this approach is motivated by the Cram\'er-Lundberg approximation); others use jump processes, as in our case.

The common ground of the majority of those papers is the complete information setting. However, in the real world the insurer has only a partial information at disposal. In fact, only the claims occurrences (times and sizes) are directly observable. Precisely, the claims intensity is a mathematical object and it is required by all the risk models, but its realizations are not observed by economic agents (as mentioned in \cite[Chapter 2]{grandell:risk}). In practice, the insurer relies on an estimation, which is based on the information at disposal. The same applies to the claim sizes distribution, which is estimated by the accident realizations. In \cite{liangbayraktar:optreins} we recognize a noteworthy attempt to introduce a partial information framework. At first, they introduce a stochastic factor $Y$ which influences the risk process. As discussed in \cite{grandell:risk}, this external driver $Y$ represents any environmental alteration reflecting on risk fluctuations (for a discussion in a complete information context see also \cite{BC:IME2019}). Then, they suppose that $Y$ is not observable. Consequently, the intensity is unobservable itself. Since $Y$ is a finite-state Markov chain in that work, the classical Hamilton-Jacobi-Bellman (HJB) approach works well after the reduction to an equivalent problem with complete information (this result is achieved by means of the filtering techniques).

In our paper we study the optimal reinsurance problem under partial information. The insurer wishes to maximize the expected exponential utility of the terminal wealth, using the information at disposal. We propose a risk model with claim arrival intensity and claim sizes distribution affected by an unobservable environmental stochastic factor $Y$. More specifically, the loss process is a marked point process with dual predictable projection dependent on $Y$, extending the Cram\`er-Lundberg model (where a Poisson process with constant intensity is used). In contrast to \cite{liangbayraktar:optreins}, here $Y$ is a general Markov process (including finite-state Markov chains, diffusions and jump-diffusions as special cases). Using filtering techniques with marked point process observations, the original problem is reduced to an equivalent stochastic control problem under complete information. Since the filter process turns out to be infinite-dimensional, the classical HJB method does not apply and we use a Backward Stochastic Differential Equation (BSDE)-based approach. Precisely, we characterize the value process and the optimal reinsurance strategy in terms of the solution to a BSDE, whose existence and uniqueness are ensured under suitable hypotheses. This is a well established approach in the financial literature, indeed several papers (see e.g. \cite{karoui:1997}, \cite{ceci:2011DEF}, \cite{lim2011}, \cite{ceci:2011} and \cite{ceci:2012IJTAF} and references therein) deal with stochastic optimization problems in finance by means of BSDEs. The recent book \cite{delong2013} applies BSDE techniques also to actuarial problems, extending the classical mathematical tools in this field. 

%In our paper we consider a stochastic factor which is a general Markov process $Y$ (including finite-state Markov chains, diffusions and jump-diffusions as special cases). Both the intensity and the claim sizes distribution are affected by $Y$. This leads to a stochastic control problem in presence of an infinite-dimensional filter process, hence the classical HJB approach does not apply. As a consequence, we use a Backward Stochastic Differential Equation (BSDE)-based approach to solve our problem, that is the maximization of the expected exponential utility of the terminal wealth, with the surplus invested in a risk-free asset with rate $R>0$. Precisely, we characterize the value process as the solution to a BSDE, whose existence and uniqueness are ensured under suitable hypotheses.  This is a well established approach in the financial literature, in fact several papers (see e.g. \cite{karoui:1997}, \cite{ceci:2011DEF}, \cite{lim2011}, \cite{ceci:2011} and \cite{ceci:2012IJTAF} and references therein) deal with stochastic optimization problems in finance by means of BSDEs.

Moreover, we model the insurance gross risk premium and the reinsurance premium as stochastic processes. Clearly, they are adapted to the filtration which represents the restricted information, since the insurance and the reinsurance companies choose the premium based on the information at disposal. 

Another important peculiarity of our work is that we consider a generic reinsurance contract, which is characterized by the self-insurance function (which represents the insurer's retained losses). Hence the retention level is chosen in the interval $[0,I]$, with $I\in(0,+\infty]$. Evidently, the proportional and the excess-of-loss optimal policies can be derived as special cases.

Finally, we allow the insurer to invest the surplus in a risk-free asset with interest rate $R>0$. The absence of a financial market with a risky asset is not restrictive. In fact, the existing literature (e.g. \cite{BC:IME2019}) have shown that the optimal reinsurance strategy only depends on the risk-free asset, even in presence of a risky asset, under the standard assumption of independence between the financial and the insurance markets. In this case, the investment strategy can be eventually determined using one of the well known results on this topic.

The paper is organized as follows: in Section \ref{section:formulation} the model is formulated and the problem is introduced. In particular, the original problem with partial information is reduced to an equivalent problem with complete information via filtering with marked point process observations. Some details about filtering results can be found in the Appendix. In Section \ref{section:BSDE} we derive a complete characterization of the value process in terms of a solution to a BSDE, whose existence and uniqueness are discussed. In addition to this, we prove the existence of an optimal reinsurance strategy under suitable conditions. Section \ref{section:optimal_reinsurance} is devoted to investigate the structure of the optimal reinsurance strategy. Finally, in Section \ref{section:comparison} we investigate some properties of the optimal reinsurance strategy, such as the Markovianity with respect to the filter process, and we discuss some relevant examples. In particular, the effect of the safety loading is analyzed and a comparison with the optimal strategy under full information is illustrated.

%\subsection*{Frequently used notation}

%--------------------------------------------------------------------------------
%	PROBLEM FORMULATION
%--------------------------------------------------------------------------------

\section{Problem formulation}
\label{section:formulation}

\subsection{Model formulation}

Let $T>0$ be a finite time horizon and assume that $(\Omega,\mathcal{G},\mathbb{P},\mathbb{G})$ is a complete probability space endowed with a filtration $\mathbb{G}\doteq\{\mathcal{G}_t\}_{t\in [0,T]}$ satisfying the usual conditions. This filtration represents all the achievable information, so that the knowledge of $\mathbb{G}$ means full information. We assume that the insurance market is influenced by an external driver $Y=\{Y_t\}_{t\in [0,T]}$, modeled as a c\`adl\`ag Markov process with infinitesimal generator $\mathcal{L}^Y$. Clearly, the sigma-algebra $\mathbb{F}^Y$ generated by $Y$ is included in $\mathbb{G}$, that is $\mathcal{F}^Y_t\subseteq\mathcal{G}_t$ $\forall t\in[0,T]$. For instance, $Y$ could be a finite-state Markov chain, a diffusion process, a jump-diffusion and so on.
%\begin{equation}
%\label{eqn:Y_gen}
%\mathcal{L}^Yf(t,y) = b(t,y) \frac{\partial{f}}{\partial{y}}(t,y)+\frac{1}{2}\gamma(t,y)^2\frac{\partial^2{f}}{\partial{y^2}}(t,y),
%\qquad f\in\mathcal{C}^{1,2}((0,T)\times\mathbb{R}).
%\end{equation}
This stochastic factor represents any environmental alteration reflecting on risk fluctuations. In practice, as suggested by Grandell, J. (see~\cite{grandell:risk}, Chapter 2), in automobile insurance $Y$ may describe road conditions, weather conditions (foggy days, rainy days, \dots), traffic volume, and so on (see also~\cite{BC:IME2019}).

The insurer's losses are described by the double sequence $\{(T_n,Z_n)\}_{n=1,\dots}$, where
\begin{itemize}
\item $\{T_n\}_{n\ge1}$ is a sequence of $\mathbb{G}$-stopping times such that $T_n<T_{n+1}$ $\mathbb{P}$-a.s. $\forall n\ge1$, representing the claims arrival times;
\item $\{Z_n\}_{n\ge1}$ is a sequence of $\mathcal{G}_{T_n}$-measurable and $(0,+\infty)$-valued random variables, which are the claims amounts.
\end{itemize}

The corresponding random measure $m(dt,dz)$ is given by
\begin{equation}
\label{eqn:random_measure}
m(dt,dz) \doteq \sum_{n\ge1}{\delta_{(T_n,Z_n)}(dt,dz)}\mathbbm{1}_{\{T_n\le T\}},
\end{equation}
where $\delta_{(t,z)}$ denotes the Dirac measure at point $(t,z)$. The marked point process $m(dt,dz)$ is characterized by the next hypotheses.

We propose a risk model with both the claims intensity and the claim sizes distribution affected by the stochastic factor $Y$. For this purpose, we use the following assumption.

\begin{assumption}
\label{ass:randommeasure}
Given a measurable function $\lambda(t,y):[0,T]\times\mathbb{R}\to(0,+\infty)$, let us define the $\mathbb{G}$-predictable process $\{\lambda_t\doteq\lambda(t,Y_{t^-})\}_{t\in [0,T]}$. Suppose that there exists a constant $\Lambda>0$ such that
\begin{equation}
\label{eqn:intensity_bounded}
0< \lambda(t,y) \le \Lambda \qquad \forall(t,y)\in[0,T]\times\mathbb{R}.
\end{equation}
In addition to this, suppose that there exists a probability transition kernel $F_Z(t,y,dz)$ from $([0,T]\times\mathbb{R},\mathcal{B}([0,T])\otimes\mathcal{B}(\mathbb{R}))$ into $([0,+\infty),\mathcal{B}([0,+\infty)))$ such that
\begin{equation}
\label{eqn:z2_finite}
\mathbb{E}\biggl[\int_0^T\int_0^{+\infty} z^2\,F_Z(t,Y_t,dz)\,dt\biggr]< +\infty.
\end{equation}

Then we assume that $m(dt,dz)$ admits the following $\mathbb{G}$-dual predictable projection:
\begin{equation}
\label{eqn:Gdual_projection}
\nu(dt,dz) = \lambda_tF_Z(t,Y_{t^-},dz)\,dt,
\end{equation}
i.e. for every nonnegative, $\mathbb{G}$-predictable and $[0,+\infty)$-indexed process $\{H(t,z)\}_{t\in[0,T]}$ we have that
\[
\mathbb{E}\biggl[\int_0^T\int_0^{+\infty}H(t,z)\,m(dt,dz)\biggr]=\mathbb{E}\biggl[\int_0^T\int_0^{+\infty} H(t,z)\,\lambda_tF_Z(t,Y_t,dz)\,dt\biggr].
\]
\end{assumption}

We will denote by $\mathbb{F}\doteq\{\mathcal{F}_t\}_{t\in[0,T]}$ the filtration generated by $m(dt,dz)$, that is
\begin{equation}
\label{eqn:filtrationF}
\mathcal{F}_t = \sigma\{ m((0,s]\times A), s\le t, A\in\mathcal{B}([0,+\infty)) \}.
\end{equation}
Using the marked point processes theory%
\footnote{For details on this topic see~\cite{bremaud:pointproc}.},
it is possible to obtain a precise interpretation of $\{\lambda_t\}_{t\in [0,T]}$ and $F_Z(t,y,dz)$ separately.

Let us denote by $N_t = m((0,t] \times [0,+\infty)) = \sum_{n\ge1} \mathbbm{1}_{\{T_n\le t\}}$ the claims arrival process, which counts the number of occurred claims. According to the definition of dual predictable projection, choosing $H(t,z) = H_t$ with $\{H_t\}_{t\in[0,T]}$ any nonnegative $\mathbb{G}$-predictable process, we get that
\begin{equation*}
\mathbb{E}\biggl[\int_0^T\int_0^{+ \infty}H_t \,m(dt,dz)\biggr]=\mathbb{E}\biggl[\int_0^T  H_t\,dN_t \biggr]  = \mathbb{E}\biggl[\int_0^TH_t \,\lambda_t\,dt\biggr] ,
\end{equation*}
i.e., $\{N_t\}_{t\in [0,T]}$ is a point process with $\mathbb{G}$-intensity $\{\lambda_t\}_{t\in[0,T]}$. 

Moreover, $F_Z(t, y,dz)$ can be interpreted as the conditional distribution of the claim sizes given the knowledge of the stochastic factor.
\begin{prop}
$\forall n=1, \dots$ and $\forall A\in \mathcal{B}([0,+\infty))$
\[
\mathbb{P}[Z_n \in A\mid \mathcal{G}_{T_n^-} ] = \int_A F_Z(T_n,Y_{T_n^-},dz)
=\mathbb{P}[Z_n \in A\mid \mathcal{F}^Y_{T_n^-} ]  \qquad  \mathbb{P}\text{-a.s.},
\]
where $\mathcal{G}_{T_n^-}$ is the strict past of the $\sigma$-algebra until time $T_n$:
$$\mathcal{G}_{T_n^-} := \sigma\{ A \cap \{ t < T_n\},  A \in \mathcal{G}_t,  t \in [0,T]\},$$
and $\mathcal{F}_{T_n^-}$ is defined similarly.
\end{prop}
\begin{proof}
See~\cite[Proposition 1]{BCrisks}.
%\cite[Proposition 2.3]{cg:2006}
\end{proof}

We define the cumulative claims up to time $t\in[0,T]$ as follows:
\begin{equation}
\label{eqn:C_t_def}
C_t = \sum_{n=1}^{N_t}Z_n = \int_0^t\int_0^{+\infty} z\,m(ds,dz).
\end{equation}

Let us observe that our model formulation is able to fit some well known risk models (for which the reader can refer to \cite{rolski:insurancefin} or \cite{schmidli:2018risk}).

\begin{example}[Cram\'er-Lundberg Risk Model]
\label{example:cramer}
If we consider a constant intensity $\lambda(t,y)=\lambda$ and a distribution function $F_Z(t, y,dz)=F_Z(dz)$, then we obtain the classical Cram\'er-Lundberg risk model.
\end{example}

\begin{example}[Markov Modulated Risk Model]
\label{example:MMRM}
Suppose that the stochastic factor $Y$ is a continuous time irreducible Markov process with a finite state space ${\cal S} = \{1, \dots, M\}$, with $M\ge2$. Taking $\lambda(y)$ and $F_Z(y,dz)$ we obtain the so called Markov Modulated Risk Model. Equivalently, we can associated $M$ constants $\{\lambda_i\}_{i=1,\dots,M}$ and distribution functions $\{F_Z^i(dz)\}_{i=1,\dots,M}$ to each state of $Y$. Correspondingly, we can define $M$ independent classical risk models (as in Example \ref{example:cramer}), with loss processes $\{C^i\}_{i=1,\dots,M}$ such that Eq. \eqref{eqn:C_t_def} becomes
\[
C_t = \int_0^t\sum_{i\in\cal{S}}\mathbbm{1}_{Y_t=i}dC^i_s .
\]
Eventually, without loss of generality we could assume that
\[
\lambda_i\int_0^{+\infty}zF_Z^i(dz) \le \lambda_j\int_0^{+\infty}zF_Z^j(dz) \qquad \forall i,j\in\{1, \dots, M\}, i<j.
\]
\end{example}

%--------------------------------------------------------------------------------
%	PROBLEM STATEMENT
%--------------------------------------------------------------------------------

\subsection{Problem statement}

In the rest of the paper we suppose that the insurer is not able to get access to the complete information $\mathbb{G}$. In contrast, at any time $t\in[0,T]$ she is allowed to observe only these objects:
\begin{itemize}
\item the occurred claims times, i.e. the jump times of $m(dt,dz)$ up to time $t$;
\item the occurred claims size, i.e. the marks of $m(dt,dz)$ up to time $t$.
\end{itemize}
More formally, the information flow at insurer's disposal is described by $\mathbb{F}\subseteq\mathbb{G}$, defined in Eq. \eqref{eqn:filtrationF}. In fact, in risk theory the claims intensity is a mathematical object and its realizations are not directly observed by economic agents (see \cite[Chapter 2]{grandell:risk}). In practice the insurer relies on an estimation of the intensity and this is based on the information at disposal, which is made of the accidents realizations. This is the basic idea behind the filtering techniques. We further extend this concept to the claim sizes distribution, which is included in the filter. That is, the insurer estimates the intensity and the size distribution at the same time.

In this framework we suppose that the gross risk premium rate $\{c_t\}_{t\in [0,T]}$ is an $\mathbb{F}$-predictable nonnegative process (the insurance company chooses the premium based on the information flow)
 such that
\begin{equation}
\label{eqn:cpremium_int}
\mathbb{E}\biggl[\int_0^T c_t\,dt\biggr]< +\infty.
\end{equation}

The insurer can subscribe a generic reinsurance contract with retention level $u\in[0,I]$, where $I>0$ (eventually $I=+\infty$), transferring part of her risks to the reinsurer. More precisely, we model the retained losses using a generic self-insurance function $g(z,u)\colon[0,+\infty)\times[0,I]\to[0,+\infty)$ which characterizes the reinsurance agreement.
\begin{remark}
\label{g_properties}
Here we recall some useful properties of the self-insurance function according to the classical risk theory%
\footnote{See \cite[Chapter 4]{schmidli:control} or \cite{schmidli:2018risk}.}:
\begin{itemize}
\item $g$ is increasing in both the variables $z,u$; moreover, it is continuous in $u\in[0,I]$;
\item $g(z,u)\le z$ $\forall u\in[0,I]$, because the retained loss is always less or equal than the claim amount;
\item $g(z,0)=0$ $\forall z\in[0,+\infty)$, because $u=0$ is the full reinsurance;
\item $g(z,I)=z$ $\forall z\in[0,+\infty)$,  because $u=I$ is the null reinsurance.
\end{itemize}
\end{remark}

Our general formulation includes standard reinsurance agreements as special cases.

\begin{example}
\label{example:contracts}
Under a proportional reinsurance the insurer transfers a percentage $u$ of any future loss, hence $I=1$ and
\[
g(z,u)=uz, \qquad u\in[0,1].
\]
Under an excess-of-loss policy the reinsurer covers all the losses which overshoot a threshold $u$, that is $I=+\infty$ and
\[
g(z,u)=z\land u, \qquad u\in[0,+\infty).
\]
\end{example}

In order to continuously buy a reinsurance agreement, the primary insurer pays a reinsurance premium $\{q^u_t\}_{t\in [0,T]}$, which is an $\mathbb{F}$-predictable 
%\foo tnote{We assume that both the insurance and the reinsurance companies have access to the same information.} 
nonnegative process satisfying the following assumption.

\begin{assumption}(Reinsurance premium)
\label{def:reinsurance_premium}
We assume that the reinsurance premium admits the following representation:
\[
q^u_t(\omega) = q(t,\omega,u) \qquad \forall (t,\omega,u)\in[0,T]\times\Omega\times[0,I],
\]
for a given function $q(t,\omega,u)\colon[0,T]\times\Omega\times[0,I]\to[0,+\infty)$ continuous and decreasing in $u$, with partial derivative $\frac{\partial q(t,\omega,u)}{\partial u}$ continuous in $u$.
In the rest of the paper $\frac{\partial q(t,\omega,0)}{\partial u}$ and $\frac{\partial q(t,\omega,I)}{\partial u}$ are interpreted as right and left derivatives, respectively.\\
In the sequel it is natural to assume that
\[
q(t,\omega,I)=0 \qquad \forall (t,\omega)\in[0,T]\times\Omega,
\]
because a null protection is not expensive. Moreover, we prevent the insurer from gaining a risk-free profit by assuming that
\[
q(t,\omega,0)>c_t \qquad \forall (t,\omega)\in[0,T]\times\Omega.
\]
The reinsurance premium associated with a dynamic reinsurance strategy $\{u_t\}_{t\in [0,T]}$ will be denoted by $\{q^u_t\}_{t\in [0,T]}$ as well, with the obvious meaning depending on context.\\
Finally, we assume the following integrability condition:
\[
\mathbb{E}\biggl[ \int_0^T q^0_t\,dt \biggr] <+\infty.
\]
\end{assumption}
As mentioned above, the premia are $\mathbb{F}$-predictable. This is a natural assumption in our context, because all the economic agents decisions are based on the common available information, which is described by $\mathbb{F}$.

Under these hypotheses, the surplus (or reserve) process associated with a given reinsurance strategy $\{u_t\}_{t\in [0,T]}$ is described by the following SDE:
\begin{equation}
\label{eqn:surplus_process}
dR^u_t = \bigl[c_t-q^u_t\bigl]dt - \int_0^{+\infty}g(z,u_t)\,m(dt,dz), \qquad R^u_0= R_0\in\mathbb{R}^+.
\end{equation}

Furthermore, we allow the insurer to invest her surplus in a risk-free asset (bond or bank account) with constant rate $R>0$. As a consequence, the insurer's wealth $\{X^u_t\}_{t\in [0,T]}$ associated with a given strategy $\{u_t\}_{t\in [0,T]}$ follows this dynamic:
\begin{equation}
\label{eqn:wealth_proc}
dX^u_t = dR^u_t+ RX^u_t\,dt, \qquad X^u_0=R_0\in\mathbb{R}^+.
\end{equation}

\begin{remark}
It can be verified that the solution to the SDE~\eqref{eqn:wealth_proc} is given by
\begin{equation}
\label{eqn:wealth_sol}
X^u_t = R_0e^{Rt} + \int_0^t e^{R(t-r)}\bigl[c_r-q^u_r\bigr]\,dr
-\int_0^t\int_0^{+\infty} e^{R(t-r)}g(z,u_r)\,m(dr,dz).
\end{equation}
\end{remark}

Now we are ready to formulate the optimization problem of an insurance company which subscribes a reinsurance contract with a dynamic retention level $\{u_t\}_{t\in [0,T]}$. The objective is to maximize the expected utility of the terminal wealth:
\[
\sup_{u\in\mathcal{U}}{\mathbb{E}\bigl[U(X_T^u)\bigr]},
\]
where $U:\mathbb{R}\to[0,+\infty)$ is the utility function representing the insurer's preferences and $\mathcal{U}$ the class of admissible strategies (see Definition \ref{def_U} below). Since only a partial information is available to the insurer and it is described by the filtration $\mathbb{F}$, the retention level $u$ turns out to be an $\mathbb{F}$-predictable process and a control problem with partial information arises.\\

We focus on CARA (\textit{Constant Absolute Risk Aversion}) utility functions, whose general expression is given by
\[
U(x) = 1-e^{-\eta x}, \qquad x\in\mathbb{R},
\]
where $\eta>0$ is the risk-aversion parameter. This utility function is highly relevant in economic science and in particular in insurance theory, in fact it is commonly used for reinsurance problems (e.g. see \cite{BC:IME2019} and references therein).\\
In this case our maximization problem reads as
\begin{equation}
\label{eqn:maximization_problem}
\sup_{u\in\mathcal{U}}{\mathbb{E}\bigl[1-e^{-\eta X^u_T}\bigr]}.
\end{equation}

\begin{definition}[Admissible strategies]
\label{def_U}
We denote by $\mathcal{U}$ the set of all the admissible strategies, which are all the $\mathbb{F}$-predictable processes $\{u_t\}_{t\in[0,T]}$ with values in $[0,I]$ such that
\[
\mathbb{E}\bigl[e^{-\eta X^u_T}\bigr]< +\infty.
\]
When we want to restrict the controls to the time interval $[t,T]$, we will use the notation $\mathcal{U}_t$.
\end{definition}

%Using the dynamic programming principle we consider the associated dynamic problem, which consists in finding the optimal strategy $\{u_s\}_{s\in [t,T]}$ for the following optimization problem, given the information available at the time $t\in[0,T]$:
%\[
%\esssup_{u\in\mathcal{U}_t}{\mathbb{E}\biggl[U(X^u_{t,x}(T))\mid \mathcal{F}_t\biggr]},
%\qquad t\in[0,T],
%\]
%with $\{X^u_{t,x}(s)\}_{s\in [t,T]}$ denoting the solution to \eqref{eqn:wealth_proc} with initial state $(t,x)$.
%For the sake of notation simplicity, we study the corresponding minimizing problem for the function $e^{-\eta x}$:
%\begin{equation}
%\label{eqn:dynamic_min_pb}
%J_t \doteq\essinf_{u\in\mathcal{U}_t}{\mathbb{E}\biggl[e^{-\eta X^u_{t,x}(T)}\mid \mathcal{F}_t\biggr]}.
%\end{equation}

We can show that $\mathcal{U}$ is a nonempty class under suitable hypotheses.

\begin{assumption}
\label{ass:verification}
The following conditions hold good:
\begin{gather}
\label{eqn:p1_1}
\mathbb{E}[e^{2\eta e^{RT}C_T}]<+\infty,\\
\label{eqn:p1_2}
\mathbb{E}[e^{2\eta e^{RT}\int_0^T e^{-Rt}q^0_t\,dt}]<+\infty.
\end{gather}
\end{assumption}

\begin{prop}
\label{prop:admissibleproc}
Under Assumption \ref{ass:verification} every $\mathbb{F}$-predictable process $\{u_t\}_{t\in[0,T]}$ with values in $[0,I]$ is admissible, that is $u\in\mathcal{U}$.
\end{prop}
\begin{proof}
By our hypotheses, taking into account that $q^u_t\le q^0_t$ $\forall t\in[0,T]$ and $\forall u\in\mathcal{U}$ (see Assumption \ref{def:reinsurance_premium}) and using the well-known inequality $ab\le\frac{1}{2}(a^2+b^2)$ $\forall a,b\in\mathbb{R}$, we have that
\[
\begin{split}
\mathbb{E}\bigl[e^{-\eta X^u_T}\bigr]
&=\mathbb{E}\biggl[ e^{-\eta e^{RT}R_0}e^{-\eta\int_0^T 
e^{R(T-t)}(c_t-q^u_t)\,dt}e^{\eta \int_0^T\int_0^{+\infty}e^{R(T-t)}g(z,u_t)\,m(dt,dz)} \biggr]\\
&\le\mathbb{E}\biggl[ e^{\eta \int_0^T e^{R(T-t)}q^0_t\,dt} e^{\eta e^{RT}\int_0^T\int_0^{+\infty}z\,m(dt,dz)} \biggr]\\
&\le \frac{1}{2} \biggl(\mathbb{E}\bigl[ e^{2\eta e^{RT} \int_0^T e^{-Rt}q^0_t\,dt} \bigr] 
+ \mathbb{E}\bigl[ e^{2\eta e^{RT}C_T} \bigr]\biggr)<+\infty,
\end{split}
\]
hence Definition \ref{def_U} is satisfied.
\end{proof}

A sufficient condition for Eq. \eqref{eqn:p1_1} can be obtained by the following lemma with the choice $p=2$.
\begin{lemma}
\label{lemma:exp_C_t_finite}
Let $p>0$ and assume that there exists an integrable function $\Phi_p:[0,T]\to(0,+\infty)$ such that
\begin{equation}
\label{eqn:F_Z_integrable}
\int_0^{+\infty}\bigl(e^{p\eta e^{RT}z}-1\bigr)F_Z(t,y,dz) 
\le \Phi_p(t) \qquad \forall (t,y)\in[0,T]\times\mathbb{R}.
\end{equation}
Then the following property holds good:
\begin{equation}
\label{eqn:exp_C_t_finite}
\mathbb{E}[e^{p\eta e^{RT}C_t}]< +\infty \qquad\forall t\in[0,T].
\end{equation}
\end{lemma}
\begin{proof}
Since $\{C_t\}_{t\in[0,T]}$ is a pure-jump process (see Eq. \eqref{eqn:C_t_def}), we have that
\begin{align*}
e^{p\eta e^{RT}C_t} &=e^{p\eta e^{RT}C_0} +\sum_{s\le t}\biggl(e^{p\eta e^{RT}C_s}-e^{p\eta e^{RT}C_{s^-}}\biggr)\\
&=1 +\sum_{s\le t}e^{p\eta e^{RT}C_{s^-}}\biggl(e^{p\eta e^{RT}\Delta C_s}-1\biggr)\\
&=1 +\int_0^te^{p\eta e^{RT}C_{s^-}}\int_0^{+\infty}\biggl(e^{p\eta e^{RT}z}-1\biggr)m(ds,dz).
\end{align*}
Taking the expectation, by \eqref{eqn:Gdual_projection},~\eqref{eqn:intensity_bounded} and \eqref{eqn:F_Z_integrable} we get that
\begin{align*}
\mathbb{E}[e^{p\eta e^{RT}C_t}] &=1 +\mathbb{E}\biggl[ \int_0^te^{p\eta e^{RT}C_{s^-}}\int_0^{+\infty}\biggl(e^{p\eta e^{RT}z}-1\biggr)\lambda_sF_Z(s,Y_s,dz)\,ds \biggr]\\
&\le 1 + \Lambda\int_0^t\mathbb{E}\bigl[e^{p\eta e^{RT}C_s}\bigr]\Phi_p(s)\,ds.
\end{align*}
Applying Gronwall's lemma we finally obtain that
\[
\mathbb{E}[e^{p\eta e^{RT}C_t}]\le e^{\Lambda\int_0^t \Phi_p(s)ds}.
%\mathbb{E}[e^{\eta e^{RT}C_t}]\le 1+ \Lambda\int_0^t \Phi(s)e^{\int_s^t \Phi(r)dr}\,ds
\]
\end{proof}

\begin{remark}
Let us denote by $m_Z(k) \doteq\mathbb{E}[e^{kZ}]$, $k\in\mathbb{R}$, the moment generating function of $Z$. Assuming $F_Z(t, y,dz)=F_Z(dz)$ as in Example \ref{example:cramer}, the condition \eqref{eqn:F_Z_integrable} is equivalent to
\[
m_Z(p\eta e^{RT})<+\infty.
\]
In particular, in view of Lemma \ref{lemma:exp_C_t_finite}, $m_Z(2\eta e^{RT})<+\infty$ implies Eq. \eqref{eqn:p1_1}.\\
As special cases we may consider the following distribution functions:
\begin{itemize}
\item if $Z\sim\Gamma(\alpha,\zeta)$ we have that $m_Z(k)=\frac{\Gamma(\alpha)}{(\zeta-k)^2}$ $\forall k<\zeta$, where $\Gamma$ denotes the gamma function; hence Eq. \eqref{eqn:p1_1} is fulfilled for any $\zeta>2\eta e^{RT}$;
\item if $Z$ is exponentially distributed, then $Z\sim\Gamma(1,\zeta)$ and hence the same condition $\zeta>2\eta e^{RT}$ applies;
\item if $Z$ has a truncated normal distribution on the interval $[0,+\infty)$, then 
\[
m_Z(k)=e^{\mu k + \frac{\sigma^2k^2}{2}}\frac{1-\mathcal{N}(-\frac{\mu}{\sigma}-\sigma k)}{1-\mathcal{N}(-\frac{\mu}{\sigma})} \qquad \forall k>0,
\]
where $\mathcal{N}$ denotes the standard normal distribution function.
\end{itemize}
%Pareto solo $k<0$ (che non ci serve), log-normale non esiste
\end{remark}

\begin{remark}
Let us consider the special case of complete information. We denote by $\{S^u_t\}_{t\in[0,T]}$ the insurer's wealth in a full information framework, that is
\[
S^u_t = R_0e^{Rt} + \int_0^t e^{R(t-r)}\bigl[\bar{c}_r-\bar{q}^u_r\bigr]\,dr
-\int_0^t\int_0^{+\infty} e^{R(t-r)}g(z,u_r)\,m(dr,dz),
\]
where the $\mathbb{G}$-predictable processes $\{\bar{c}_t\}_{t\in[0,T]}$ and $\{\bar{q}_t\}_{t\in[0,T]}$ denote the insurance and the reinsurance premium, respectively. In order to simplify the comparison, the full and the partial information frameworks are defined in a similar way. $\mathcal{U}^G$ denotes the class of admissible strategies and it is defined as in Definition \ref{def_U}, replacing $\mathbb{F}$ with $\mathbb{G}$ and $X^u_t$ with $S^u_t$. Under Assumption \ref{ass:verification}, as in Proposition \ref{prop:admissibleproc}, we can prove the admissibility of every $\mathbb{G}$-predictable process. Hence, since any $\mathbb{F}$-predictable process is also $\mathbb{G}$-predictable, we get $\mathcal{U}\subseteq\mathcal{U}^G$.
We take the same insurance premia $c_t=\bar{c}_t$ and reinsurance premia $q^u_t=\bar{q}^u_t$ $\forall u\in\mathcal{U}$. 
In this simple context, we can readily get that
\[
\mathbb{E}\bigl[e^{-\eta X^u_T}\bigr] = \mathbb{E}\bigl[e^{-\eta S^u_T}\bigr]
\qquad \forall u\in\mathcal{U},
\]
and, as a consequence,
\[
\inf_{u\in\mathcal{U}^G}{\mathbb{E}\bigl[e^{-\eta S^u_T}\bigr]} 
\le \inf_{u\in\mathcal{U}}{\mathbb{E}\bigl[e^{-\eta S^u_T}\bigr]}
= \inf_{u\in\mathcal{U}}{\mathbb{E}\bigl[e^{-\eta X^u_T}\bigr]}.
\]
In words, the complete information allows the insurer to improve her result. However, we point out that such an expected result is no longer easy to prove in general (for example when the premia do not coincide). \\
In Section \ref{section:comparison} we will compare the optimal strategies under partial information with those under complete information in some special cases.
%we will show that for well known reinsurance agreements (as in Example \ref{example:contracts}) and under premium calculation principles (as in Example \ref{example:premia} below) the optimal retention level with partial information is always less or equal than the one with complete information.}
\end{remark}

%--------------------------------------------------------------------------------
%	REDUCTION TO COMPLETE OBSERVATION
%--------------------------------------------------------------------------------

\subsection{Reduction to a complete information problem}

In the previous subsection we have introduced the partially observable problem. In order to study it, we need to reduce it to an equivalent problem with complete information.
This can be achieved by deriving the compensator $m^\pi(dt,dz)$ of the random measure given in Eq. \eqref{eqn:random_measure}, that is the insurer's loss process, with respect to its internal filtration $\mathbb{F}$, which represents the information at disposal to the insurance and the reinsurance companies. 
%We will reduce our model to a full information model by deriving the compensator $m^\pi(dt,dz)$ of the random measure $m(dt,dz)$ with respect its internal filtration, $\mathbb{F}$, which represents the available information at disposal to the insurance and the reinsurance companies.
This result can be obtained by solving a filtering problem with marked point process observations.  It is well known that the filter,  that is the conditional distribution  of $Y_t$ given the $\sigma$-algebra $\mathcal{F}_t$, for any $t\in[0,T]$,  provides the best mean-squared estimate of the unobservable stochastic factor $Y$  from the available information.  
Precisely, the filter is the $\mathbb{F}$-adapted c\`adl\`ag process $\{\pi_t(f)\}_{t\in[0,T]}$ taking values in the space of probability measures on $\mathbb{R}$ defined by
\[
\pi_t(f) = \mathbb{E}[f(t,Y_t)\mid\mathcal{F}_t],
\]
for any measurable function $f\colon[0,T]\times\mathbb{R}\to\mathbb{R}$ such that $\mathbb{E}[\abs{f(t,Y_t)}]< +\infty$ $\forall t\in[0,T]$.

By applying \cite[Proposition 2.2]{cc:AAP2012}, we can derive $m^\pi(dt,dz)$.

\begin{lemma}
The random measure $m(dt,dz)$ given in~\eqref{eqn:random_measure} has $\mathbb{F}$-dual predictable projection $m^\pi(dt,dz)$ given by 
$\pi_{t^-}(\lambda F_Z(dz))\,dt,$ that is, the following expression holds for any $A\in\mathcal{B}([0,+\infty))$
\begin{equation}
\label{eqn:dual_projection}
m^\pi(dt,A) \doteq \pi_{t^-}(\lambda(t,\cdot) F_Z(t,\cdot,A))\,dt,
\end{equation} 
where $\pi_{t}(\lambda(t,\cdot) F_Z(t,\cdot,A))=  \mathbb{E}[\lambda(t,Y_t) F_Z(t,Y_t,A)\mid\mathcal{F}_t]$ and  $\pi_{t^-}$ denotes the left version of the process $\pi_{t}$.
\end{lemma}

\begin{remark}
\label{remark:dualprojection}
By definition of dual predictable projection, for every nonnegative, $\mathbb{F}$-predictable and $[0,+\infty)$-indexed process $\{H(t,z)\}_{t\in[0,T]}$ we have that
\[\mathbb{E}\biggl[\int_0^T\int_0^{+\infty} H(t,z)\,m(dt,dz)\biggr] =\]

\[
\mathbb{E}\biggl[\int_0^T\int_0^{+\infty} H(t,z)\,\lambda_tF_Z(t,Y_t,dz)\,dt\biggr]
=\mathbb{E}\biggl[\int_0^T\int_0^{+\infty} H(t,z)\,\pi_{t^-}(\lambda F_Z(dz))\,dt\biggr].
\]
\end{remark}

By Remark \ref{remark:dualprojection} we can rewrite the classical premium calculation principles adapting them to our dynamic and partially observable context via the filter process%
\footnote{See \cite{young:premium_princ} for the original formulation in a static framework.}.

\begin{example}[Premium calculation principles]
\label{example:premia}
Under the expected value principle, the expected revenue covers the expected losses plus a profit which is proportional to the expected losses:
\begin{align}
c_t &= (1+\theta_i)\int_0^{+\infty} z\,\pi_{t^-}(\lambda F_Z(dz)), \notag\\
\label{eqn:evp}
q^u_t &= (1+\theta)\int_0^{+\infty} (z-g(z,u_t))\,\pi_{t^-}(\lambda F_Z(dz)),
\end{align}
where $\theta>\theta_i>0$ represent the safety loadings.\\
Under the variance premium principle, the expected gain is proportional to the variance of the losses instead:
\begin{align}
c_t &= \int_0^{+\infty} z\,\pi_{t^-}(\lambda F_Z(dz)) 
+ \theta_i \int_0^{+\infty} z^2\,\pi_{t^-}(\lambda F_Z(dz)) , \notag\\
\label{eqn:vp}
q^u_t &= \int_0^{+\infty} (z-g(z,u_t))\,\pi_{t^-}(\lambda F_Z(dz))
+ \theta \int_0^{+\infty}(z-g(z,u_t))^2\,\pi_{t^-}(\lambda F_Z(dz)),
\end{align}
for some safety loadings $\theta>\theta_i>0$. Observe that in these examples the premium at time $t$  depends on the estimate of  the compensator of the loss process given the available information immediately before time $t$, that is 
$\pi_{t^-}(\lambda F_Z(dz)) dt$.

A formal derivation of these premium calculation rules in a dynamic context can be found in \cite{BC:IME2019} and \cite{BCrisks}.\\
\end{example}

Filtering problems with marked point process observations have been widely investigated in the literature, 
see \cite{bremaud:pointproc} and more recently \cite{cg:2006} and \cite{ceci:2006}. See also \cite{cc:AAP2012} and \cite{cc:AMO2014} for jump-diffusion observations.
Here, starting from the existing literature, we derive an explicit formula for the filter under general assumptions on the stochastic factor $Y$. Precisely, we assume $Y$ to be a c\`adl\`ag Markov process, but we do not assign any specific dynamics to $Y$. More details can be found in Appendix.

Let us denote by $\mathcal{L}^Y$ the Markov generator of $Y$ with domain $\mathcal{D}^Y$,
 that is for every function $f \in \mathcal{D}^Y \subseteq C_b([0,T]\times\mathbb{R})$
$$f(t,Y_t) = f(t_0,y_0) + \int_{t_0}^t \mathcal{L}^Y f(s,Y_s) ds + M^Y_t, \quad t \in [0,T],$$
for some $\mathbb{F}^Y$-martingale  $\{M^Y_t\}_{t\in[0,T]}$ and $(t_0, y_0) \in [0,T]\times \mathbb{R}$. 

\begin{assumption}
\label{Y}
We assume the following standard hypotheses:
\begin{itemize}
\item for any initial value $(t_0, y_0) \in [0,T]\times \mathbb{R}$ the martingale problem%
\footnote{See \cite{ek:1986} for details about martingale problems.} 
for the operator $\mathcal{L}^Y$ is well posed on the space of c\`adl\`ag trajectories (this is true, for instance, when $Y$ is the unique strong solution of a SDE for any initial values $(t_0, y_0) \in [0,T]\times\mathbb{R} $);
\item $\mathcal{L}^Yf\in C_b([0,T]\times\mathbb{R})$ for any $f \in \mathcal{D}^Y $;
\item $\mathcal{D}^Y $ is an algebra dense in $C_b([0,T]\times\mathbb{R})$.
\end{itemize}
\end{assumption}

For simplicity, we assume no common jump times between $Y$ and $m(dt,dz)$ (we should specify the dynamic for $Y$ to remove such a simplification). 
% In the sequel we recall the main results that apply in our framework. The filter can be characterized as the unique strong solution of the so called Kushner-Stratonovich equation. We refer to \cite{ceci:2006} and \cite{cc:AAP2012} for a detailed proof. 

\begin{prop} \label{new}
Under Assumption \ref{Y}, letting $y_0 \in \mathbb{R}$ be a fixed initial value for $Y$ at time $t=0$, 
the filter $\pi$ can be obtained by the following recursive procedure%with respect to the jump times $\{T_n\}_{n\geq 0}$ ($T_0:= 0$), that is 
\begin{itemize}
\item $\pi_0(f) = f(0,y_0)$, $\forall t \in (0, T_1)$
$$\pi_{t}(f) =  { E[ f(t, Y_t) e^{- \int_0^t \lambda(r, Y_r) dr} | Y_0=y_0] \over E [ e^{- \int_0^t \lambda(r, Y_r) dr} | Y_0=y_0] };$$

\item at a jump time $T_n$, $n \geq 1$:

\begin{equation} \label{jump} \pi_{T_n}(f) = W(T_n, \pi_{T_n^-}, Z_n) \doteq {d \pi_{T_n^-}(\lambda F_Z f) \over d \pi_{T_n^-}(\lambda F_Z ) }(Z_n), \end{equation}
where ${d \pi_{t^-}(\lambda F_Z f) \over d \pi_{t^-}(\lambda F_Z ) }(z)$  denotes the Radon-Nikodym derivative of the measure $\pi_{t^-}(\lambda F_Z (dz)f)$  with respect to $ \pi_{t^-}(\lambda F_Z(dz))$;

\item between two consecutive jump times, $t \in (T_{n}, T_{n+1})$, $n \geq 1$: 
\[
\pi_{t}(f) = { E_{n}[ f(t, Y_t) e^{- \int_s^t \lambda(r, Y_r) dr} ] |_{s= T_{n}}\over E_{n}[ e^{- \int_s^t \lambda(r, Y_r) dr}] |_{s= T_{n}}},
\]
where $E_{n}$ denotes the conditional  expectation given  the distribution $Y_{T_{n}} $ equal to $\pi_{T_{n}}$.
\end{itemize}

\end{prop}

\begin{proof}
The results are derived in Appendix.
\end{proof}

Similarly to \cite[Section 3.3]{cg:2006}, by Proposition \ref{new} we can write a recursive algorithm to approximate the filter. We conclude the section with some special cases. The following results are discussed in Appendix.

\begin{remark}[Known jump size distribution and unknown intensity]
\label{remark:filterapp}
In the special case where $F_Z(t, y,dz)=F_Z(dz)$, that is,  the insurance company has complete knowledge on the claim size distribution and partial information on the claim arrival intensity.  Eq. \eqref{jump} reduces to
\begin{equation}
\label{eqn:filter1}
\pi_{T_n}(f) = W(T_n,\pi_{T_n^-}) = {\pi_{T_n^-}(\lambda f) \over  \pi_{T_n^-}(\lambda)},
\end{equation}
see Example \ref{example:A1} in Appendix.

If $Y$ takes values in a discrete  set ${\cal S} = \{1,2, \dots\}$, defining the functions $f_i(y):=\mathbbm{1}_{y= i}$, $i \in {\cal S}$, the filter is completely described via the knowledge of 
$\pi_{t}(i):= \pi_{t}(f_i) = P(Y_t = i \mid\mathcal{F}_t)$, $i \in {\cal S}$, because for every function $f$ 
we have that $\pi_i(f) = \sum_{i \in  {\cal S}} f(i) \pi_{t}(i).$ Eq. \eqref{jump} reads as 
\begin{equation}
\label{eqn:filter2}
\pi_{T_n}(i)  = { d(\lambda({T_n},i)  F_Z({T_n}, i, dz)\pi_{{T_n^-}}(i) ) \over 
 d( \sum_{j \in  {\cal S}} \lambda({T_n},j)  F_Z({T_n}, j, dz) \pi_{{T_n^-}}(j)) }(Z_n),
\end{equation}
 which, in the special case $F_Z(t, y,dz)=F_Z(dz)$, simplifies to
\begin{equation}
\label{eqn:filter3}
\pi_{T_n}(i) = W_i(T_n,\pi_{T_n^-}) \doteq { \lambda({T_n},i) \pi_{T_n^-}(i)
\over  \sum_{j \in  {\cal S}} \lambda({T_n},j) \pi_{T_n^-}(j) },
\end{equation}
see Example \ref{example:A3} in Appendix.
\end{remark}

\begin{remark}[Markov Modulated Risk Model with infinitely many states]
\label{remark:MMRM}
If $Y$ takes values in a discrete  set ${\cal S} = \{1,2, \dots\}$, the random measure $m(dt,dz)$ in~\eqref{eqn:random_measure} 
has $\mathbb{F}$-dual predictable projection given by
\[
m^\pi(dt, dz) = \sum_{i \in  {\cal S}}  \pi_{t^-}(i) \lambda(t,i) F_Z(t,i,dz) \,dt.
\]
In particular, under the Markov Modulated Risk Model (see Example \ref{example:MMRM}) we get 
\[
m^\pi(dt, dz) = \sum_{i =1} ^{M}  \pi_{t^-}(i) \lambda(i) F^i_Z(dz) \,dt.
\]
and  by Eq. \eqref{eqn:filter3}
\begin{equation}
\label{eqn:filterWlast}
\pi_{T_n}(i) = W_i(\pi_{T_n^-}) 
= { \lambda_i \pi_{T_n^-}(i)
\over  \sum_{j =1}^{M} \lambda_j \pi_{T_n^-}(j) },
\end{equation}
see Example \ref{example:A2} in Appendix.
 \cite{liangbayraktar:optreins} consider this case under the assumption that the claim distribution for any state $i=1, \dots M$ admits density, that is $F^i_Z(dz)  = f_i (z)dz$.
\end{remark}

 %--------------------------------------------------------------------------------
%	THE BSDE APPROACH
%--------------------------------------------------------------------------------

\section{The BSDE approach}
\label{section:BSDE}

As usual in stochastic control problems, we introduce the dynamic problem associated to \eqref{eqn:maximization_problem}. For the sake of notational simplicity, we study the corresponding minimization problem for the function $e^{-\eta x}$. Precisely, for any admissible control $u\in\mathcal{U}$ let us define the Snell envelope:
\begin{equation}
\label{eqn:snell_envelope}
J^u_t \doteq \essinf_{\bar{u}\in\mathcal{U}(t,u)}
{\mathbb{E}\biggl[e^{-\eta X^{\bar{u}}_T}\mid \mathcal{F}_t\biggr]} ,
\end{equation}
where $\mathcal{U}(t,u)$ denotes the class $\mathcal{U}$ restricted to the controls $\bar{u}$ such that $\bar{u}_s=u_s$ $\forall s\le t$, for a given arbitrary control $u\in\mathcal{U}$.\\
Let us introduce the discounted wealth $\{\bar{X}^u_t\doteq e^{-Rt}X^u_t\}_{t\in[0,T]}$, that is
\begin{equation}
\label{eqn:Xbar}
\bar{X}^u_t = R_0 + \int_0^t e^{-Rs}\bigl[c_s-q^u_s\bigr]\,ds
-\int_0^t\int_0^{+\infty} e^{-Rs}g(z,u_s)\,m(ds,dz), \qquad t\in[0,T].
\end{equation}
Then, by Eq. \eqref{eqn:wealth_sol} we get
\begin{equation}
\label{eqn:J^u}
J^u_t=e^{-\eta \bar{X}^u_t e^{RT}}V_t,
\end{equation}
where we define the value process
\begin{equation}
\label{eqn:V_t}
V_t \doteq \essinf_{\bar{u}\in\mathcal{U}_t}{\mathbb{E}\biggl[e^{-\eta e^{RT}(\bar{X}^{\bar{u}}_T-\bar{X}^{\bar{u}}_t)}\mid \mathcal{F}_t\biggr]},
\end{equation}
with $\mathcal{U}_t$ denoting the class of admissible controls restricted to the time interval $[t,T]$ (see Definition \ref{def_U}).

By Eqs.~\eqref{eqn:Xbar} and \eqref{eqn:J^u} it is easy to show that
\begin{align}
J^u_t &= e^{-\eta (\bar{X}^u_t-\bar{X}^I_t) e^{RT}} e^{-\eta \bar{X}^I_t e^{RT}}V_t \notag\\
\label{eqn:diff_equality}
&=e^{\eta (\bar{X}^I_t-\bar{X}^u_t) e^{RT}} J^I_t,
\end{align}
and
\begin{equation}
\label{eqn:Vt_JI}
V_t=e^{\eta \bar{X}^I_t e^{RT}}J^I_t,
\end{equation}
where $J^I_t$ denotes the Snell envelope associated to $u=I$ (null reinsurance).\\

The goal of this section is to dynamically characterize the value process by using a BSDE-based approach. The BSDE method works well in non-Markovian settings, where the classical stochastic control approach based on the Hamilton-Jacobi-Bellman equation does not apply. Several papers (see e.g. \cite{karoui:1997}, \cite{ceci:2011DEF}, \cite{lim2011} and references therein) deal with stochastic optimization problems in finance by means of BSDEs.  For insurance applications the reader can refer to the recent textbook \cite{delong2013}.  Moreover, this approach is also well suited to solve stochastic control problems under partial information in presence of an infinite-dimensional filter process (see e.g. \cite{ceci:2011} and \cite{ceci:2012IJTAF}, where partially observed power utility maximization problems in financial markets are solved by applying this approach).

\begin{prop}
\label{prop:VJfinite}
Under Assumption \ref{ass:verification} we have that
\begin{equation}
\label{eqn:4bis}
\mathbb{E}[(\sup_{t\in[0,T]}{J^I_t})^2] <+\infty .
\end{equation}
\end{prop}
\begin{proof}
%The following inequalities hold true:
%\begin{gather}
%\label{eqn:V_finite}
%0\le V_t \le \mathbb{E}[e^{\eta e^{RT}(C_T-C_t)}\mid\mathcal{F}_t]
%\qquad \mathbb{P}\text{-a.s.} \quad \forall t\in[0,T],\\
%\label{eqn:J_finite}
%0\le J^I_t \le \mathbb{E}[e^{\eta e^{RT}C_T}\mid\mathcal{F}_t]
%\qquad \mathbb{P}\text{-a.s.} \quad \forall t\in[0,T].
%\end{gather}
By Eq. \eqref{eqn:Xbar} for $u=I$ (null reinsurance) we have that
\[
\bar{X}^I_t = R_0 + \int_0^t e^{-Rs}c_s\,ds - \int_0^t\int_0^{+\infty} e^{-Rs}z\,m(ds,dz).
\]
By definition of $V_t$ (see Eq. \eqref{eqn:V_t}), since $u=I\in\mathcal{U}$
\[
\begin{split}
0 \le V_t &\le \mathbb{E}[e^{-\eta e^{RT}(\bar{X}^I_T-\bar{X}^I_t)}\mid \mathcal{F}_t]\\
&\le \mathbb{E}[e^{\eta e^{RT}(C_T-C_t)}\mid\mathcal{F}_t]
\qquad \mathbb{P}\text{-a.s.} \quad \forall t\in[0,T].
\end{split}
\]
Analogously, by definition of $J^I_t$ (see Eq. \eqref{eqn:J^u}) we immediately get
\[
\begin{split}
0 \le J^I_t &= e^{-\eta \bar{X}^I_t e^{RT}}V_t\\
&\le e^{\eta C_t e^{RT}} \mathbb{E}[e^{\eta e^{RT}(C_T-C_t)}\mid\mathcal{F}_t]\\
&=  \mathbb{E}[e^{\eta e^{RT}C_T}\mid\mathcal{F}_t]
\qquad \mathbb{P}\text{-a.s.} \quad \forall t\in[0,T].
\end{split}
\]
It follows that
\[
J^I_t \le \mathbb{E}[e^{\eta e^{RT}C_T}\mid\mathcal{F}_t] \doteq m_t,
\]
where $\{m_t\}_{t\in[0,T]}$ is an $\mathbb{F}$-martingale. By Doob's martingale inequality, we have that
\[
\begin{split}
\mathbb{E}[(\sup_{t\in[0,T]}{J^I_t})^2] &\le \mathbb{E}[(\sup_{t\in[0,T]}{m_t})^2] \\
&\le 4\,\mathbb{E}[m_T^2] \\
&= 4\,\mathbb{E}[e^{2\eta e^{RT}C_T}] <+\infty .
\end{split} 
\]
\end{proof}

Our aim is to prove that the process $\{J^I_t\}_{t\in[0,T]}$ solves a BSDE driven by the compensated jump measure $m(dt, dz)-\pi_{t^-}(\lambda F_Z(dz))\,dt$. In order to derive this BSDE, we need the following additional hypotheses. Strengthening the assumptions is useful for deriving the BSDE at this stage, but in the Verification Theorem (see Theorem \ref{theorem:verification} below) we will come back to the weaker Assumption \ref{ass:verification}.

\begin{assumption}
\label{ass:bellman}
The following conditions are satisfied:
\begin{gather}
\label{eqn:A21}
\mathbb{E}[e^{2\eta p e^{RT}C_T}]<+\infty \qquad \forall p\ge1,\\
\label{eqn:A22}
\mathbb{E}[e^{2\eta p e^{RT}\int_0^T e^{-Rs}q^0_s\,ds}]<+\infty \qquad \forall p\ge1 .
\end{gather}
\end{assumption}

\begin{remark}
Under the classical premium calculation principles \eqref{eqn:evp} and \eqref{eqn:vp}, Eq. \eqref{eqn:A22} is fulfilled if we take the claim sizes distribution $F_Z(t,y,dz)=F_Z(dz)$ such that
\[
\int_0^{+\infty} z^2\,F_Z(dz) <+\infty,
\]
In fact, in this case $q^0_t$ is a bounded process and hence Eq. \eqref{eqn:A22} is clearly satisfied.
\end{remark}

\begin{prop}[Bellman's optimality principle]
\label{prop:Bellman_optimality}
Under Assumption \ref{ass:bellman} the following statements hold good:
\begin{enumerate}
\item $\{J^u_t\}_{t\in[0,T]}$ is an $\mathbb{F}$-sub-martingale for any $u\in\mathcal{U}$;
\item $\{J^{u^*}_t\}_{t\in[0,T]}$ is an $\mathbb{F}$-martingale if and only if $u^*\in\mathcal{U}$ is an optimal control.
\end{enumerate}
\end{prop}
\begin{proof}
By \cite[Prop. 4.1]{lim2011}, the result is valid if $\forall u\in\mathcal{U}$ and $\forall p\ge1$
\[
\mathbb{E}[\sup_{s\in[t,T]}{e^{-\eta p X^u_{t,x}(s)}}]<+\infty \qquad \forall(t,x)\in[0,T]\times\mathbb{R},
\]
where $\{X^u_{t,x}(s)\}_{s\in[t,T]}$ denotes the solution to Eq. \eqref{eqn:wealth_proc} with initial condition $(t,x)\in[0,T]\times\mathbb{R}$. We observe that
\[
\begin{split}
e^{-\eta p X^u_{t,x}(s)} &\le e^{\eta p e^{Rs}\int_t^s e^{-Rr}q^u_r\,dr}e^{\eta p e^{Rs}C_s}\\
&\le \frac{1}{2} \bigl(e^{2\eta p e^{Rs}\int_t^s e^{-Rr}q^u_r\,dr}+e^{2\eta p e^{Rs}C_s}\bigr)
\qquad \mathbb{P}\text{-a.s.} \quad \forall t\in[0,T],
\end{split}
\]
hence $\forall(t,x)\in[0,T]\times\mathbb{R}$ we get
\[
\mathbb{E}[\sup_{s\in[t,T]}{e^{-\eta p X^u_{t,x}(s)}}] \le \frac{1}{2} \bigl(
\mathbb{E}[e^{2\eta p e^{RT}\int_0^T e^{-Rs}q^0_s\,ds}] 
+ \mathbb{E}[e^{2\eta p e^{RT}C_T}] \bigr)<+\infty.
\]
\end{proof}

\begin{remark}
Under Assumption \ref{ass:bellman} we can apply Bellman's optimality principle (see Proposition \ref{prop:Bellman_optimality}). Since $u=I\in\mathcal{U}$, $\{J^I_t\}_{t\in[0,T]}$ is an $\mathbb{F}$-sub-martingale.
Consequently, by Doob-Meyer decomposition and the martingale representation theorems%
\footnote{E.g. see~\cite[Theorem T8]{bremaud:pointproc}.}, it admits the following expression:
\begin{equation}
\label{eqn:J_doobmeyer}
J^I_t = \int_0^t\int_0^{+\infty}\Gamma(s,z)\bigl(m(ds, dz)-\pi_{s^-}(\lambda F_Z(dz))\,ds)\bigr) + A_t,
\end{equation}
where by \eqref{eqn:4bis} $\Gamma(t,z)$ is a $[0,+\infty)$-indexed $\mathbb{F}$-predictable process such that
\[
\mathbb{E}\biggl[\int_0^T\int_0^{+\infty}\abs{\Gamma(s,z)}^2\,\pi_{s^-}(\lambda F_Z(dz))\,ds\biggr]< +\infty,
\]
and $\{A_t\}_{t\in [0,T]}$ is an increasing $\mathbb{F}$-predictable process such that $\mathbb{E}[\int_0^T A_s^2\,ds]<+\infty$.
\end{remark}

\begin{lemma}[Snell envelope decomposition]
\label{J_differential}
Under Assumption \ref{ass:bellman}, for any $u\in\mathcal{U}$ the Snell envelope $\{J^u_t\}_{t\in [0,T]}$ admits the following representation:
\begin{equation}
\label{eqn:Ju_diff}
dJ^u_t =dM^u_t + e^{\eta (\bar{X}^I_t-\bar{X}^u_t) e^{RT}}\bigl[A_t - f(t,\Gamma(t,z),J^I_t,u_t)\bigr]\,dt,
\end{equation}
where
\begin{align}
M^u_t &\doteq \int_0^t e^{\eta (\bar{X}^I_{s^-}-\bar{X}^u_{s^-}) e^{RT}} \int_0^{+\infty}\Gamma(s,z)e^{-\eta e^{R(T-s)}(z-g(z,u_s))}\bigl(m(ds, dz)-\pi_{s^-}(\lambda F_Z(dz))\,ds)\bigr) \notag\\
\label{eqn:M^u_martingale}
&+\int_0^t J^I_{s^-}e^{\eta (\bar{X}^I_{s^-}-\bar{X}^u_{s^-}) e^{RT}}\int_0^{+\infty} \biggl(e^{-\eta e^{R(T-s)}(z-g(z,u_s))}-1\biggr)\bigl(m(ds, dz)-\pi_{s^-}(\lambda F_Z(dz))\,ds)\bigr)
\end{align}
is an $\mathbb{F}$-martingale and
\begin{align}
f(t,\Gamma(t,z), J^I_t,u_t)
&\doteq - J^I_{t^-}\eta e^{R(T-t)}q^u_t\notag\\
\label{eqn:driver_f}
&-\int_0^{+\infty} \bigl(J^I_{t^-}+\Gamma(t,z)\bigr)\biggl(e^{-\eta e^{R(T-t)}(z-g(z,u_t))}-1\biggr)\pi_{t^-}(\lambda F_Z(dz)).
\end{align}
\end{lemma}
\begin{proof}
Since $J^u_t=e^{\eta (\bar{X}^I_t-\bar{X}^u_t) e^{RT}} J^I_t$ by Eq. \eqref{eqn:diff_equality}, we focus on the computation of the latter term.
By the product rule for stochastic integrals we get that
\begin{multline}
\label{eqn:product_rule_u}
d(e^{\eta (\bar{X}^I_t-\bar{X}^u_t) e^{RT}} J^I_t) =
e^{\eta (\bar{X}^I_{t^-}-\bar{X}^u_{t^-}) e^{RT}}\,d J^I_t
+ J^I_{t^-}\,d(e^{\eta (\bar{X}^I_t-\bar{X}^u_t) e^{RT}}) \\
+d\biggl(\sum_{s\le t}{\Delta J^I_s\Delta e^{\eta (\bar{X}^I_s-\bar{X}^u_s) e^{RT}}}\biggr).
\end{multline}
Let us evaluate~\eqref{eqn:product_rule_u} item by item. Using the expression \eqref{eqn:J_doobmeyer} we can easily obtain the first term. By Eq. \eqref{eqn:Xbar} we get
\begin{equation}
\label{eqn:discountedXI}
\bar{X}^I_t-\bar{X}^u_t = \int_0^t e^{-Rs}q^u_s\,ds - \int_0^t\int_0^{+\infty} e^{-Rs}(z-g(z,u_s))\,m(ds,dz).
\end{equation}
Hence by It\^o's formula we have that
\[
\begin{split}
d(e^{\eta (\bar{X}^I_t-\bar{X}^u_t) e^{RT}}) &= \eta e^{RT}e^{\eta (\bar{X}^I_t-\bar{X}^u_t) e^{RT}}e^{-Rt}q^u_t\,dt\\
&+d\biggl(\sum_{s\le t}{e^{\eta (\bar{X}^I_{s^-}-\bar{X}^u_{s^-}) e^{RT}}\biggl(e^{\eta e^{RT} \bigl((\bar{X}^I_s-\bar{X}^u_s)- (\bar{X}^I_{s^-}-\bar{X}^u_{s^-})\bigr)}-1\biggr)}\biggr)\\
&= \eta e^{RT}e^{\eta (\bar{X}^I_t-\bar{X}^u_t) e^{RT}}e^{-Rt}q^u_t\,dt\\
&+e^{\eta (\bar{X}^I_{t^-}-\bar{X}^u_{t^-}) e^{RT}}\int_0^{+\infty} \biggl(e^{-\eta e^{R(T-t)}(z-g(z,u_t))}-1\biggr)m(dt,dz).
\end{split}
\]

By the last equation we also find out that
\[
d\biggl(\sum_{s\le t}{\Delta J^I_s\Delta e^{\eta (\bar{X}^I_{s}-\bar{X}^u_{s}) e^{RT}}}\biggr)
=e^{\eta (\bar{X}^I_{t^-}-\bar{X}^u_{t^-}) e^{RT}}\int_0^{+\infty} \Gamma(t,z)\biggl(e^{-\eta e^{R(T-t)}(z-g(z,u_t))}-1\biggr)m(dt,dz).
\]
Let us come back to~\eqref{eqn:product_rule_u}. We have just obtained that
\[
\begin{split}
d(e^{\eta (\bar{X}^I_t-\bar{X}^u_t) e^{RT}} J^I_t) &=
e^{\eta (\bar{X}^I_{t^-}-\bar{X}^u_{t^-}) e^{RT}}\biggl[\int_0^{+\infty}\Gamma(t,z)\bigl(m(dt, dz)-\pi_{t^-}(\lambda F_Z(dz))\,dt)\bigr) + dA_t\biggr]\\
&+ J^I_{t^-}\eta e^{RT}e^{\eta (\bar{X}^I_t-\bar{X}^u_t) e^{RT}}e^{-Rt}q^u_t\,dt\\
&+ J^I_{t^-}e^{\eta (\bar{X}^I_{t^-}-\bar{X}^u_{t^-}) e^{RT}}\int_0^{+\infty} \biggl(e^{-\eta e^{R(T-t)}(z-g(z,u_t))}-1\biggr)m(dt,dz)\\
&+e^{\eta (\bar{X}^I_{t^-}-\bar{X}^u_{t^-}) e^{RT}}\int_0^{+\infty} \Gamma(t,z)\biggl(e^{-\eta e^{R(T-t)}(z-g(z,u_t))}-1\biggr)m(dt,dz).
\end{split}
\]
After some calculations, we rewrite it as
\[
\begin{split}
&d(e^{\eta (\bar{X}^I_t-\bar{X}^u_t) e^{RT}} J^I_t) \\
& \quad =e^{\eta (\bar{X}^I_{t^-}-\bar{X}^u_{t^-}) e^{RT}} \int_0^{+\infty}\Gamma(t,z)e^{-\eta e^{R(T-t)}(z-g(z,u_t))}\bigl(m(dt, dz)-\pi_{t^-}(\lambda F_Z(dz))\,dt)\bigr)\\
& \quad +J^I_{t^-}e^{\eta (\bar{X}^I_{t^-}-\bar{X}^u_{t^-}) e^{RT}}\int_0^{+\infty} \biggl(e^{-\eta e^{R(T-t)}(z-g(z,u_t))}-1\biggr)\bigl(m(dt, dz)-\pi_{t^-}(\lambda F_Z(dz))\,dt)\bigr)\\
& \quad +e^{\eta (\bar{X}^I_{t^-}-\bar{X}^u_{t^-}) e^{RT}}\,dA_t
+ J^I_{t^-}\eta e^{RT}e^{\eta (\bar{X}^I_t-\bar{X}^u_t) e^{RT}}e^{-Rt}q^u_t\,dt\\
& \quad +e^{\eta (\bar{X}^I_{t^-}-\bar{X}^u_{t^-}) e^{RT}}\int_0^{+\infty} \bigl(J^I_{t^-}+\Gamma(t,z)\bigr)\biggl(e^{-\eta e^{R(T-t)}(z-g(z,u_t))}-1\biggr)\pi_{t^-}(\lambda F_Z(dz))\,dt.
\end{split}
\]
%\[
%\begin{split}
%d(e^{\eta (\bar{X}^I_t-\bar{X}^u_t) e^{RT}} J^I_t) &=
%e^{\eta (\bar{X}^I_{t^-}-\bar{X}^u_{t^-}) e^{RT}} \int_0^{+\infty}\Gamma(t,z)e^{-\eta e^{R(T-t)}(z-g(z,u_t))}\bigl(m(dt, dz)-\pi_{t^-}(\lambda F_Z(dz))\,dt)\bigr)\\
%&+J^I_{t^-}e^{\eta (\bar{X}^I_{t^-}-\bar{X}^u_{t^-}) e^{RT}}\int_0^{+\infty} \biggl(e^{-\eta e^{R(T-t)}(z-g(z,u_t))}-1\biggr)\bigl(m(dt, dz)-\pi_{t^-}(\lambda F_Z(dz))\,dt)\bigr)\\
%&+e^{\eta (\bar{X}^I_{t^-}-\bar{X}^u_{t^-}) e^{RT}}\,dA_t
%+ J^I_{t^-}\eta e^{RT}e^{\eta (\bar{X}^I_t-\bar{X}^u_t) e^{RT}}e^{-Rt}q^u_t\,dt\\
%&+e^{\eta (\bar{X}^I_{t^-}-\bar{X}^u_{t^-}) e^{RT}}\int_0^{+\infty} \bigl(J^I_{t^-}+\Gamma(t,z)\bigr)\biggl(e^{-\eta e^{R(T-t)}(z-g(z,u_t))}-1\biggr)\pi_{t^-}(\lambda F_Z(dz))\,dt.
%\end{split}
%\]
By definition of $\{M^u_t\}_{t\in [0,T]}$ and $\{f(t,\Gamma(t,z), J^I_t,u_t)\}_{t\in [0,T]}$ (see Eqs.~\eqref{eqn:M^u_martingale} and \eqref{eqn:driver_f}, respectively),
we obtain the expression \eqref{eqn:Ju_diff}.

In order to complete the proof, we need to show that $\{M^u_t\}_{t\in [0,T]}$ is an $\mathbb{F}$-martingale for any $u\in\mathcal{U}$, that is
\begin{gather*}
%\label{eqn:M_mg1}
\mathbb{E}\biggl[ \int_0^T e^{\eta (\bar{X}^I_{s^-}-\bar{X}^u_{s^-}) e^{RT}} \int_0^{+\infty}\abs{\Gamma(s,z)}e^{-\eta e^{R(T-s)}(z-g(z,u_s))}\pi_{s^-}(\lambda F_Z(dz))\,ds \biggr]<+\infty,\\
%\label{eqn:M_mg2}
\mathbb{E}\biggl[ \int_0^T J^I_{s^-} e^{\eta (\bar{X}^I_{s^-}-\bar{X}^u_{s^-}) e^{RT}}\int_0^{+\infty} \abs*{e^{-\eta e^{R(T-s)}(z-g(z,u_s))}-1} \pi_{s^-}(\lambda F_Z(dz))\,ds \biggr]<+\infty.
\end{gather*}
In the rest of the proof $C>0$ denotes a generic constant.
By Remark \ref{remark:dualprojection} and Eq. \eqref{eqn:discountedXI} we observe that
\[
\begin{split}
&\mathbb{E}\biggl[ \int_0^T e^{\eta (\bar{X}^I_{s^-}-\bar{X}^u_{s^-}) e^{RT}} \int_0^{+\infty}\abs{\Gamma(s,z)}e^{-\eta e^{R(T-s)}(z-g(z,u_s))}\lambda_s F_Z(s,Y_s,dz)\,ds \biggr] \\
&\le \mathbb{E}\biggl[ e^{\eta e^{RT} \int_0^T e^{-Rs}q^0_s\,ds} \int_0^T \int_0^{+\infty}\abs{\Gamma(s,z)} \lambda_s F_Z(s,Y_s,dz)\,ds \biggr] \\
%&\le \mathbb{E}\biggl[ e^{\eta e^{RT} \int_0^T e^{-Rs}q^0_s\,ds} \sqrt{\int_0^T \int_0^{+\infty}\abs{\Gamma(s,z)}^2 \lambda_s F_Z(s,Y_s,dz)\,ds} \biggr] \\
&\le C\, \mathbb{E}\biggl[ e^{2\eta e^{RT} \int_0^T e^{-Rs}q^0_s\,ds}\biggr] + C\, \mathbb{E}\biggl[\int_0^T \int_0^{+\infty}\abs{\Gamma(s,z)}^2 \pi_{s^-}(\lambda F_Z(dz))\,ds \biggr]<+\infty.
\end{split}
\]
Now let us evaluate the second expectation. By Remark \ref{remark:dualprojection}, Eq. \eqref{eqn:discountedXI} and Eq. \eqref{eqn:4bis}
\[
\begin{split}
&\mathbb{E}\biggl[ \int_0^T J^I_{s^-} e^{\eta (\bar{X}^I_{s^-}-\bar{X}^u_{s^-}) e^{RT}}\int_0^{+\infty} \abs*{e^{-\eta e^{R(T-s)}(z-g(z,u_s))}-1} \lambda_s F_Z(s,Y_s,dz)\,ds \biggr] \\
&\le \Lambda \mathbb{E}\biggl[\int_0^T J^I_{s^-} e^{\eta e^{RT} \int_0^T e^{-Rr}q^0_r\,dr}\,ds \biggr] \\
&\le C \biggl(\mathbb{E}\biggl[\int_0^T \abs{J^I_{s^-}}^2 \,ds \biggr]+\mathbb{E}\bigl[ e^{2\eta e^{RT} \int_0^T e^{-Rs}q^0_s\,ds} \bigr]\biggr)<+\infty.
\end{split}
\]
%\[
%\begin{split}
%&\mathbb{E}\biggl[ \int_0^T J^I_s e^{\eta (\bar{X}^I_s-\bar{X}^u_s) e^{RT}}\int_0^{+\infty} \abs*{e^{-\eta e^{R(T-s)}(z-g(z,u_s))}-1} \lambda_s F_Z(s,Y_s,dz)\,ds \biggr] \\
%&\le \Lambda \int_0^T  \mathbb{E}\biggl[ J^I_s e^{\eta e^{RT} \int_0^T e^{-Rr}q^0_r\,dr} \biggr]\,ds \\
%&\le \frac{1}{2}\Lambda T \biggl(\mathbb{E}\bigl[ e^{2\eta e^{RT}C_T} \bigr]+\mathbb{E}\bigl[ e^{2\eta e^{RT} \int_0^T e^{-Rs}q^0_s\,ds} \bigr]\biggr)<+\infty.
%\end{split}
%\]
%\[
%\begin{split}
%&\mathbb{E}\biggl[ \int_0^T J^I_s e^{\eta (\bar{X}^I_s-\bar{X}^u_s) e^{RT}}\int_0^{+\infty} \abs*{e^{-\eta e^{R(T-s)}(z-g(z,u_s))}-1} \lambda_s F_Z(s,Y_s,dz)\,ds \biggr] \\
%&\le \Lambda \mathbb{E}\biggl[ \int_0^T \mathbb{E}[e^{\eta e^{RT}C_T}\mid\mathcal{F}_s] e^{\eta e^{RT} \int_0^T e^{-Rs}q^0_s\,ds} \,ds \biggr] \\
%&\le \Lambda T\, \mathbb{E}\bigl[ e^{\eta e^{RT}C_T}e^{\eta e^{RT} \int_0^T e^{-Rs}q^0_s\,ds} \bigr]\\
%&\le \frac{1}{2}\Lambda T \biggl(\mathbb{E}\bigl[ e^{2\eta e^{RT}C_T} \bigr]+\mathbb{E}\bigl[ e^{2\eta e^{RT} \int_0^T e^{-Rs}q^0_s\,ds} \bigr]\biggr)<+\infty.
%\end{split}
%\]
\end{proof}

%\begin{remark}
%As shown in Lemma \ref{J_differential},
%\[
%dJ^u_t =dM^u_t + e^{\eta (\bar{X}^I_t-\bar{X}^u_t) e^{RT}}\bigl[A_t - f(t,\Gamma(t,z),J^I_t,u_t)\bigr]\,dt,
%\]
%where $\{M^u_t\}_{t\in [0,T]}$ is an $\mathbb{F}$-martingale such that $M^u_0=0$. In particular, this implies that
%\[
%\begin{split}
%\mathbb{E}\biggl[\int_0^T \bigl[A_t - f(t,\Gamma(t,z),J^I_t,u_t)\bigr]\,dt \biggr]
%&= \mathbb{E}[J^u_T] \\
%&= \mathbb{E}[e^{-\eta \bar{X}^u_T e^{RT}}]\\
%&\le \mathbb{E}[e^{\eta e^{RT} \int_0^T e^{-Rt}q^0_t\,dt} e^{\eta e^{RT}C_T} ] \\
%&\le \frac{1}{2}\biggl(\mathbb{E}[e^{2\eta e^{RT} \int_0^T e^{-Rt}q^0_t\,dt} ] + \mathbb{E}[e^{2\eta e^{RT}C_T} ] \biggr)< +\infty.
%\end{split}
%\]
%\end{remark}

\begin{definition}
\label{def:BSDEsol}
We introduce the following classes of stochastic processes:
\begin{itemize}
\item $\mathcal{L}^2$ is the space of c\`adl\`ag $\mathbb{F}$-adapted processes $\{\hat{J_t}\}_{t\in[0,T]}$ such that
	\begin{equation}
	\label{eqn:supJhat_finite}
	\mathbb{E}\biggl[\int_0^T \abs{\hat{J_t}}^2\,dt\biggr]<+\infty.
	\end{equation}
\item $\mathcal{\tilde{L}}^2$ is the space of $[0,+\infty)$-indexed $\mathbb{F}$-predictable processes $\{\hat{\Gamma}(t,z), z\in[0,+\infty)\}_{t\in[0,T]}$ such that 
	\begin{equation}
	\label{eqn:gamma_finite}
	\mathbb{E}\biggl[\int_0^T\int_0^{+\infty}\abs{\hat{\Gamma}(t,z)}^2\,\pi_{t^-}(\lambda F_Z(dz))\,dt\biggr]< +\infty.
	\end{equation}
\end{itemize}
\end{definition}

\begin{prop}
\label{prop:opt_control_maximiser}
Let $\{u^*_t\}_{t\in[0,T]}$ be an optimal control for the optimization problem \eqref{eqn:V_t}. Under Assumption \ref{ass:bellman} $(J^I_t,\Gamma(t,z))\in\mathcal{L}^2\times\mathcal{\tilde{L}}^2$ is a solution to the following BSDE:
\begin{equation}
\label{eqn:BSDE}
J^I_t = \xi - \int_t^T\int_0^{+\infty}\Gamma(s,z)\bigl(m(ds,dz)-\pi_{s^-}(\lambda F_Z(dz))\,ds\bigr)) -\int_t^T\esssup_{u\in\mathcal{U}}{f(s,\Gamma(s,z),J^I_s,u_s)}\,ds,
\end{equation}
where $\{f(t,\Gamma(t,z), J^I_t,u_t)\}_{t\in [0,T]}$ is defined in~\eqref{eqn:driver_f} and $\xi=e^{-\eta X^I_T}$.\\
Moreover, $f(t,\Gamma(t,z), J^I_t,u_t)$ attains its maximum in $u^*_t$, that is
\begin{equation}
\label{eqn:sup_f}
f(t,\Gamma(t,z), J^I_t,u^*_t)=\esssup_{u\in\mathcal{U}}{f(t,\Gamma(t,z), J^I_t,u_t)}.
\end{equation}
\end{prop}
\begin{proof}
For any admissible control $u\in\mathcal{U}$, by Bellman's optimality principle (Proposition \ref{prop:Bellman_optimality}) %$\{e^{-\eta \bar{X}^u_t e^{RT}}V_t\}_{t\in[0,T]}$
$\{J^u_t\}_{t\in[0,T]}$ is an $\mathbb{F}$-sub-martingale and thus by Eq. \eqref{eqn:Ju_diff} we readily get $\forall u\in\mathcal{U}$
\begin{equation}
\label{eqn:A_t_ge_f}
A_t \ge f(t,\Gamma(t,z), J^I_t,u_t) \qquad \mathbb{P}\text{-a.s.}\quad\forall t\in[0,T].
\end{equation}
Let $\{u^*_t\}_{t\in[0,T]}$ be an optimal control for the problem \eqref{eqn:V_t}. By Bellman's optimality principle $\{J^{u^*}_t\}_{t\in[0,T]}$ is an $\mathbb{F}$-martingale and by Lemma \ref{J_differential} this is true if only if
\[
A_t=f(t,\Gamma(t,z), J^I_t,u^*_t).
\]
Combining this result with \eqref{eqn:A_t_ge_f} leads to
\[
\esssup_{u\in\mathcal{U}}{f(t,\Gamma(t,z), J^I_t,u_t)}\ge f(t,\Gamma(t,z), J^I_t,u^*_t)
=A_t\ge \esssup_{u\in\mathcal{U}}{f(t,\Gamma(t,z), J^I_t,u_t)},
\]
which implies Eq. \eqref{eqn:sup_f}. Now, using the Doob-Meyer representation \eqref{eqn:J_doobmeyer}, we conclude that $(J^I_t,\Gamma(t,z))$ is a solution to \eqref{eqn:BSDE}, with the terminal condition easily derived by Eq. \eqref{eqn:J^u}.
\end{proof}

\begin{remark}
The process $\{f(t,\Gamma(t,z), J^I_t,u^*_t)\}_{t\in [0,T]}$ (see Eq. \eqref{eqn:sup_f}) is non negative. Indeed, by Eq. \eqref{eqn:driver_f} we immediately get
\[
f(t,\Gamma(t,z), J^I_t,u^*_t) \ge f(t,\Gamma(t,z), J^I_t,I) = 0.
\]
\end{remark}

%\nuovo{The next two remarks emphasize some properties of the predictable random field $\{\Gamma (t,z), z \in [0, + \infty)\}_{t \geq 0}$.}

\begin{remark}
\label{remark:gamma_unique}
The process $\{J^I_t\}_{t \geq 0}$ completely determines 
the predictable random field $\{\Gamma (t,z), z \in [0, + \infty)\}_{t \geq 0}$ outside a null set with respect to the measure $\pi_{t^-}(\lambda F_Z(dz))(\omega) \mathbb{P}(d\omega)\,dt$. 
In fact, if $(J^I_t,\Gamma(t,z))$ and $(J^I_t,\Gamma^1(t,z))$ satisfy the BSDE~\eqref{eqn:BSDE}, on the jump times of the random measure $m(dt,dz)$ we necessarily have that
\[
\Gamma(T_n,Z_n) = \Delta J^I_{T_n} = \Gamma^1(T_n,Z_n) \qquad \forall n\ge1.
\]
Hence, for any $t\in [0,T]$ and $C \in \mathcal{B}([0,+\infty))$
\begin{align}
0 &= \mathbb{E}\biggl[\int_0^t\int_C \abs{\Gamma(s,z)-\Gamma^1(s,z)}\,m(ds,dz)\biggr]\notag\\
\label{eqn:gamma_unique}
&= \mathbb{E}\biggl[\int_0^t\int_C\abs{\Gamma(s,z)-\Gamma^1(s,z)}\,\pi_{s^-}(\lambda F_Z(dz))\,ds\biggr]=0,
\end{align}
and this implies that $\Gamma(t,z)=\Gamma^1(t,z)$ $\pi_{t^-}(\lambda F_Z(dz))(\omega) \mathbb{P}(d\omega)\,dt$-a.e..
%\footnote{See also \cite[Theorem T8]{bremaud:pointproc}.}. Hence $\Gamma(t,z)$ is unique with respect to this measure.
%Now suppose that the support of $F_Z(t,y,dz)$ (defined in Assumption \ref{ass:randommeasure}) is the interval $[0,D)$ (eventually $D=+\infty$). Then by Eq. \eqref{eqn:gamma_unique} we can state that $\Gamma(t,z)=\Gamma^1(t,z)$ $\mathbb{P}$-a.s..
\end{remark}

%\begin{remark}Since $\{J^I_t\}_{t \geq 0}$ is an  $\mathbb{F}$-sub-martingale, continuous between two consecutive jump times %and $\{\Gamma (t,z),  \in [0, + \infty)\}_{t \geq 0}$  is $\mathbb{F}$-predictable, we have that$$J^I_{T_n}-J^I_{T_n-} =  \Gamma(T_n,Z_n) \geq 0 , \quad \mathbb{P}-a.s.$$Hence, by the previous remark $\Gamma(t,z)\ge 0$ $\pi_{t^-}(\lambda F_Z(dz))(\omega) \mathbb{P}(d\omega)\,dt$-a.e..Otherwise, if $ \Gamma(T_n,Z_n) <0$, over a set $B$  with positive probability, by continuity $J^I$ between two consecutive jump times of we can find an interval  such that...\begin{equation}\begin{split}0 &\le \mathbb{E}[J^I_{T_n}-J^I_{T_n-} \mid \mathcal{F}_{T_n-}] \\&= \mathbb{E}[\Gamma(T_n,Z_n) \mid \mathcal{F}_{T_n-}] \\&= \Gamma(T_n,Z_n) ,\end{split}\end{equation}with probability $1$.  \end{remark}{\bf ATTENZIONE questa dimostrazione \'e sbagliata perch\'e $\Gamma(T_n,Z_n) $ non \'e $\mathcal{F}_{T_n-}-misurabile$.}

Recalling that $V_t=e^{\eta \bar{X}^I_t e^{RT}}J^I_t$ (see Eq. \eqref{eqn:Vt_JI}), using the Bellman's optimality principle we have connected the value process \eqref{eqn:V_t} to the solution of the BSDE \eqref{eqn:BSDE}. For this purpose, we made extensive use the hypotheses included in Assumption \ref{ass:bellman}. Now a verification argument is needed. To this end, we will assume the weaker conditions given in Assumption \ref{ass:verification}.

\begin{prop}[A general Verification Theorem]
\label{prop:general_verification_th}
Under Assumption \ref{ass:verification}, let us suppose that there exists an $\mathbb{F}$-adapted process $\{D_t\}_{t\in[0,T]}$ such that
\begin{enumerate}
\item $\{D_te^{-\eta \bar{X}^u_t e^{RT}}\}_{t\in[0,T]}$ is an $\mathbb{F}$-sub-martingale for any $u\in\mathcal{U}$ and an $\mathbb{F}$-martingale for some $u^*\in\mathcal{U}$;
\item $D_T=1$.
\end{enumerate}
Then $D_t=V_t$ and $u^*$ is an optimal control.
\end{prop}
\begin{proof}
Using the terminal condition and the sub-martingale property, we have that for any $t\in[0,T]$
\[
\mathbb{E}[e^{-\eta \bar{X}^u_T e^{RT}}\mid\mathcal{F}_t]
\ge D_te^{-\eta \bar{X}^u_t e^{RT}} \qquad\forall u\in\mathcal{U},
\]
hence
\[
D_t \le \mathbb{E}[e^{-\eta e^{RT}(\bar{X}^u_T-\bar{X}^u_t)}\mid\mathcal{F}_t],
\]
which implies $D_t\le V_t$. Moreover, for $u^*\in\mathcal{U}$ we have that
\[
D_t=\mathbb{E}[e^{-\eta e^{RT}(\bar{X}^{u^*}_T-\bar{X}^{u^*}_t)}\mid\mathcal{F}_t]\ge V_t.
\]
The two inequalities imply the thesis.
\end{proof}

\begin{theorem}[Verification Theorem]
\label{theorem:verification}
Suppose that Assumption \ref{ass:verification} is fulfilled. Let $(\hat{J_t},\hat{\Gamma}(t,z))\in\mathcal{L}^2\times\mathcal{\tilde{L}}^2$ be a solution to the BSDE \eqref{eqn:BSDE} and let $u^*=\{u^*_t\}_{t\in[0,T]}$ be an $\mathbb{F}$-predictable process such that
\begin{equation}
\label{eqn:uverification}
\esssup_{u\in\mathcal{U}}{f(t,\hat{\Gamma}(t,z), \hat{J_t},u_t)} = f(t,\hat{\Gamma}(t,z), \hat{J_t},u^*_t)
\qquad \forall t\in[0,T].
\end{equation}
Then $\{D_t\doteq e^{\eta \bar{X}^I_t e^{RT}}\hat{J_t}\}_{t\in[0,T]}$ is the value process of the optimal reinsurance problem, that is $D_t=V_t$ (see Eq. \eqref{eqn:V_t}), and $u^*\in\mathcal{U}$ is an optimal control.
\end{theorem}
\begin{proof}
In view of the general Verification Theorem introduced in Proposition \ref{prop:general_verification_th}, let us consider the stochastic process $\{D_te^{-\eta \bar{X}^u_t e^{RT}}\}_{t\in [0,T]}$. Since
\[
e^{-\eta \bar{X}^u_t e^{RT}}D_t = e^{\eta (\bar{X}^I_t-\bar{X}^u_t) e^{RT}} \hat{J_t},
\]
by definition of $D_t$, using the BSDE \eqref{eqn:BSDE} and imitating the proof of Lemma \ref{J_differential}, we have that
\[
d(e^{-\eta \bar{X}^u_t e^{RT}}D_t)=d\hat{M}^u_t+e^{\eta (\bar{X}^I_t-\bar{X}^u_t) e^{RT}}
\bigl[\esssup_{w\in\mathcal{U}}{f(t,\hat{\Gamma}(t,z), \hat{J_t},w_t)}- f(t,\hat{\Gamma}(t,z), \hat{J_t},u_t)\bigr]\,dt,
\]
where $\hat{M}^u_t$ is defined in Eq. \eqref{eqn:M^u_martingale} and $f$ is given in Eq. \eqref{eqn:driver_f}  by replacing $(J^I_t,\Gamma(t,z))$ with $(\hat{J_t},\hat{\Gamma}(t,z))$. 
In order to prove that $\{\hat{M}^u_t\}_{t\in[0,T]}$ is an $\mathbb{F}$-martingale $\forall u\in\mathcal{U}$, we replicate the calculations of the proof of Lemma \ref{J_differential}. By Assumption \ref{ass:verification} we obtain that
\[
\begin{split}
&\mathbb{E}\biggl[ \int_0^T e^{\eta (\bar{X}^I_{s^-}-\bar{X}^u_{s^-}) e^{RT}} \int_0^{+\infty}\abs{\hat{\Gamma}(s,z)}e^{-\eta e^{R(T-s)}(z-g(z,u_s))}\lambda_s F_Z(s,Y_s,dz)\,ds \biggr] \\
&\le C\, \mathbb{E}\biggl[ e^{2\eta e^{RT} \int_0^T e^{-Rs}q^0_s\,ds}\biggr] + C\, \mathbb{E}\biggl[\int_0^T \int_0^{+\infty}\abs{\hat{\Gamma}(s,z)}^2 \pi_{s^-}(\lambda F_Z(dz))\,ds \biggr]<+\infty,
\end{split}
\]
where $C>0$ is a constant. Moreover, we have that
\[
\begin{split}
&\mathbb{E}\biggl[ \int_0^T \hat{J}_{s^-}e^{\eta (\bar{X}^I_{s^-}-\bar{X}^u_{s^-}) e^{RT}}\int_0^{+\infty} \abs*{e^{-\eta e^{R(T-s)}(z-g(z,u_s))}-1} \lambda_s F_Z(s,Y_s,dz)\,ds \biggr] \\
&\le \tilde{C} \, \mathbb{E}\biggl[ \int_0^T \abs{\hat{J_s}}^2\,ds \biggr] 
+ \tilde{C} \,\mathbb{E}\bigl[ e^{2\eta e^{RT} \int_0^T e^{-Rs}q^0_s\,ds} \bigr] <+\infty,
\end{split}
\]
where $\tilde{C}>0$ is a constant and the two terms are finite because of Assumption \ref{ass:verification} and condition \eqref{eqn:supJhat_finite}.

Now, it is clear that for any $u\in\mathcal{U}$
\[
\esssup_{w\in\mathcal{U}_t}{f(t,\hat{\Gamma}(t,z), \hat{J_t},w_t)} \ge f(t,\hat{\Gamma}(t,z), \hat{J_t},u_t),
\]
hence $\{e^{-\eta \bar{X}^u_t e^{RT}}D_t\}_{t\in [0,T]}$ turns out to be an $\mathbb{F}$-sub-martingale.\\
Now let us consider the $\mathbb{F}$-predictable process $\{u^*_t\}_{t\in[0,T]}$ satisfying Eq. \eqref{eqn:uverification}. In this case the previous inequality reads as an equality by definition of $u^*$, hence $\{e^{-\eta \bar{X}^{u^*}_t e^{RT}}D_t\}_{t\in [0,T]}$ is an $\mathbb{F}$-martingale. Finally,
\[
D_T = e^{\eta \bar{X}^I_T e^{RT}}\hat{J_T} = 1.
\]
As announced, our statement follows by Proposition \ref{prop:general_verification_th}.
\end{proof}

\begin{remark}
\label{existence_ustar}
Let us notice that $f$ given in Eq. \eqref{eqn:driver_f} is continuous in $u\in[0,I]$ and under Assumption \ref{ass:verification} every $\mathbb{F}$-predictable process is admissible by Proposition \ref{prop:admissibleproc}. As a consequence, an optimal control exists as long as the BSDE \eqref{eqn:BSDE} admits a solution $(\hat{J_t},\hat{\Gamma}(t,z))\in\mathcal{L}^2\times\mathcal{\tilde{L}}^2$. Precisely, there exists a measurable function $u^*(t,\omega,\gamma(\cdot),j)$, with $t\in[0,T],\omega\in\Omega,\gamma\colon[0,+\infty)\to\mathbb{R},j\in[0,+\infty)$, such that
\begin{equation}
\label{eqn:u_maximizer}
f(t,\omega,\gamma(\cdot),j,u^*(t,\omega,\gamma,j)) = \max_{u\in[0,I]}{f(t,\omega,\gamma(\cdot),j,u)}
%\quad \forall (t,\omega,\gamma(\cdot),j)\in[0,T]\times\Omega\times\mathbb{R}\times[0,+\infty).
\end{equation}
and
\[
u^*_t = u^*(t,\hat{\Gamma}(t,z),\hat{J}_{t^-})
\]
is an optimal control. This topic will be developed further in Section \ref{section:optimal_reinsurance}.
\end{remark}

%--------------------------------------------------------------------------------
%	SOLUTION TO THE BSDE
%--------------------------------------------------------------------------------

\subsection{Existence and uniqueness of solutions to BSDE \eqref{eqn:BSDE}}

In this section we deal with the solution to the BSDE \eqref{eqn:BSDE}, that provides our value process \eqref{eqn:V_t} in view of Theorem \ref{theorem:verification}. Precisely, we discuss its existence and uniqueness.

\begin{lemma}
\label{BSDE_terminal}
Suppose that Eq.~\eqref{eqn:p1_1} is fulfilled. The final condition $\xi=e^{-\eta X^I_T}$ of the BSDE \eqref{eqn:BSDE} is square-integrable.
\end{lemma}
\begin{proof}
Recalling that $q^I_t=0$ $\forall t\in[0,T]$ and $g(z,I)=z$ $\forall z\in[0,+\infty)$, by Eq. \eqref{eqn:wealth_sol} we have that
\begin{align*}
e^{-\eta X^I_T}
&=e^{-\eta R_0e^{RT}} e^{-\eta\int_0^T e^{R(T-r)}c_r\,dr}
e^{\eta\int_0^T\int_0^{+\infty} e^{R(T-r)}z\,m(dr,dz)}\\
&\le e^{\eta e^{RT}C_T} \qquad \mathbb{P}\text{-a.s.}.
\end{align*}
The statement immediately follows by Eq.~\eqref{eqn:p1_1}.
\end{proof}

Now we handle the problem of existence and uniqueness of a solution to \eqref{eqn:BSDE}.

\begin{definition}
\label{def:Theta}
For any $t\in[0,T]$ and $\omega\in\Omega$ we denote by $\Theta(t,\omega)$ the space of all the functions $\gamma\colon[0,+\infty)\to\mathbb{R}$ such that
\[
\int_0^{+\infty}\abs{\gamma(z)}\,\pi_{t^-}(\lambda F_Z(dz)) <+\infty.
\]
\end{definition}

In the sequel we use this short notation:
\[
A \doteq \Set{(t,\omega,\gamma(\cdot),j,u)\in[0,T]\times\Omega\times\Theta(t,\omega)\times[0,+\infty)\times[0,I]}.
\]
Correspondingly, we take
\[
\bar{A} \doteq \Set{(t,\omega,\gamma(\cdot),j)\in[0,T]\times\Omega\times\Theta(t,\omega)\times[0,+\infty)}.
\]

\begin{theorem}
\label{theorem:BSDE_existence}
Suppose that the following hypotheses are fulfilled:
\begin{itemize}
\item the condition~\eqref{eqn:p1_1} is fulfilled;
\item the function $q(t,\omega,u)$ given in Assumption \ref{def:reinsurance_premium} is bounded;
\end{itemize}
There exists a unique solution $(\hat{J_t},\hat{\Gamma}(t,z))\in\mathcal{L}^2\times\mathcal{\tilde{L}}^2$ which solves the BSDE~\eqref{eqn:BSDE}.
\end{theorem}
\begin{proof}
In order to apply the results of \cite{Confortola2013}, let us notice that the classes introduced in Definition \ref{def:BSDEsol} and Definition \ref{def:Theta} are equivalent to those of the cited paper, except for the absence of a parameter $\beta>0$; in fact, in our framework there is no need of this, because the compensator of the counting process $\{N_t\}_{t\in[0,T]}$ is $\{\pi_{t^-}(\lambda)\}_{t\in[0,T]}$ and it is bounded by $\Lambda>0$ (see Section \ref{section:formulation}).\\
Now let $f$ be an $\mathbb{F}$-predictable process defined on $A$ by
\begin{equation}
\label{eqn:function_f}
f(t,\omega,\gamma(\cdot),j,u)\doteq - j\eta e^{R(T-t)}q^u_t-\int_0^{+\infty} \bigl(j+\gamma(z)\bigr)\biggl(e^{-\eta e^{R(T-t)}(z-g(z,u))}-1\biggr)\pi_{t^-}(\lambda F_Z(dz)).
\end{equation}
%with $t\in[0,T],\omega\in\Omega,\gamma\in\Theta, j\in[0,+\infty),u\in[0,I]$.\\
Since $q^u_t$ is bounded by hypothesis, using the condition~\eqref{eqn:intensity_bounded} and taking $\gamma,\gamma'\in\Theta(t,\omega)$ and $j,j'\in[0,+\infty)$, we have that $f$ satisfies a Lipschitz condition uniformly in $t,\omega,u$:
\[
\begin{split}
&\abs{f(t,\omega,\gamma'(\cdot),j',u)-f(t,\omega,\gamma(\cdot),j,u)}\\
&=\abs{j\,\eta e^{R(T-t)}q^u_t
+\int_0^{+\infty} (j+\gamma(z)) \bigl(e^{-\eta e^{R(T-t)}(z-g(z,u))}-1\bigr)\pi_{t^-}(\lambda F_Z(dz))\\
&-j'\,\eta e^{R(T-t)}q^u_t
-\int_0^{+\infty} (j'+\gamma'(z)) \bigl(e^{-\eta e^{R(T-t)}(z-g(z,u))}-1\bigr)\pi_{t^-}(\lambda F_Z(dz))}\\
&\le L\,\abs{j-j'}+\abs*{\int_0^{+\infty}
\bigl(\gamma(z)-\gamma'(z)\bigr) \bigl(e^{-\eta e^{R(T-t)}(z-g(z,u))}-1\bigr)\pi_{t^-}(\lambda F_Z(dz))}\\
%&\le L\,\abs{j-j'} + L \abs*{\int_0^{+\infty}\bigl(\gamma(z)-\gamma'(z)\bigr) \pi_{t^-}(\lambda F_Z(dz))}\\
&\le L\,\abs{j-j'} + \int_0^{+\infty}\abs*{\gamma(z)-\gamma'(z)} \pi_{t^-}(\lambda F_Z(dz))\\
&\le L\,\abs{j-j'} + \Lambda \biggl(\int_0^{+\infty}\abs*{\gamma(z)-\gamma'(z)}^2 \pi_{t^-}(\lambda F_Z(dz))\biggr)^{\frac{1}{2}}
\qquad \forall t\in[0,T],\omega\in\Omega,u\in[0,I] ,
\end{split}
\]
for a suitable constant $L>0$. It can be proved that $\sup_{u\in[0,I]}{f(t,\omega,\gamma(\cdot),j,u)}$ preserves this property, in fact
\begin{align*}
&\abs*{ \sup_{u\in[0,I]}{f(t,\omega,\gamma(\cdot),j,u)} - \sup_{u\in[0,I]}{f(t,\omega,\gamma'(\cdot),j',u)} }\\
&\le \sup_{u\in[0,I]}{ \abs*{ f(t,\omega,\gamma(\cdot),j,u) 
 - f(t,\omega,\gamma'(\cdot),j',u) } }\\
&\le L\,\abs{j-j'} + \Lambda \biggl(\int_0^{+\infty}\abs*{\gamma(z)-\gamma'(z)}^2 \pi_{t^-}(\lambda F_Z(dz))\biggr)^{\frac{1}{2}}
\qquad \forall t\in[0,T],\omega\in\Omega.
\end{align*}
Further, let us observe that $f(t,\omega,0,0,u)=0$ $\forall(t,\omega,u)\in[0,T]\times\Omega\times[0,I]$ and the BSDE terminal condition is square-integrable by Lemma~\ref{BSDE_terminal}. We can deduce that Hypothesis 3.1 of \cite{Confortola2013} is fulfilled. Hypothesis 4.5 is satisfied as well, because of Remark \ref{existence_ustar}. Finally, our statement is a consequence of \cite[Theorem 3.4]{Confortola2013}.%~\cite[Theorem 3]{confortolaetal:BSDE}.
\end{proof}

Let us summarize the results of this section in the following theorem.

\begin{theorem}
\label{theorem:summary}
Suppose that Assumption \ref{ass:verification} is fulfilled and the reinsurance premium is bounded. Then $(J^I_t,\Gamma(t,z))\in\mathcal{L}^2\times\mathcal{\tilde{L}}^2$ is the unique solution of the BSDE~\eqref{eqn:BSDE}. Moreover, let $u^*(t,\omega,\gamma(\cdot),j)$ be the maximizer of Eq. \eqref{eqn:u_maximizer}, then $u^*_t = u^*(t,\Gamma(t,z),J^I_{t^-})$ is an optimal control and the value process in Eq. \eqref{eqn:V_t} admits the representation $\{V_t\doteq e^{\eta \bar{X}^I_t e^{RT}}J^I_t\}_{t\in[0,T]}$.
\end{theorem}
\begin{proof}
The BSDE~\eqref{eqn:BSDE} admits a unique solution by Theorem \ref{theorem:BSDE_existence} and the existence of an optimal control is guaranteed by Remark \ref{existence_ustar}. This in turn implies that $(J^I_t,\Gamma(t,z))\in\mathcal{L}^2\times\mathcal{\tilde{L}}^2$ is the unique solution of Eq. \eqref{eqn:BSDE} by  Proposition \ref{prop:opt_control_maximiser}. Finally, the expression of the value process is obtained by Theorem \ref{theorem:verification}.
\end{proof} 

%\begin{remark}In general, the solution to a BSDE can be regarded as the viscosity solution to the Hamilton-Jacobi-Bellman equation associated to the stochastic control problem, see e.g. REFERENCE. Moreover, in the special case of a Markov Modulated Risk Model (see Example \ref{example:MMRM}) the value process can be characterized as the solution to a generalized Hamilton-Jacobi-Bellman equation (see \cite{liangbayraktar:optreins}).\end{remark}

%--------------------------------------------------------------------------------
%	PROPORTIONAL REINSURANCE
%--------------------------------------------------------------------------------

\section{The optimal reinsurance strategy}
\label{section:optimal_reinsurance}

Eq. \eqref{eqn:u_maximizer} suggests a natural way to find an optimal strategy. This is the main topic of this section.

\begin{prop}
\label{prop:optimal_u}
Assume $g(z,u)$ differentiable in $u\in[0,I]$. Let $f$ be defined by Eq. \eqref{eqn:function_f} and suppose that it is strictly concave in $u$. Let the function $u^*(t,\omega,\gamma,j)$ be defined as follows:
\begin{equation}
\label{eqn:optimal_u}
u^*(t,\omega,\gamma(\cdot),j)=
\begin{cases}
	0 & \text{$(t,\omega,\gamma(\cdot),j)\in A_0$}
	\\
	\hat{u}(t,\omega,\gamma(\cdot),j) & \text{$(t,\omega,\gamma(\cdot),j)\in (A_0\cup A_I)^C$}
	\\
	I & \text{$(t,\omega,\gamma(\cdot),j)\in A_I$,}
\end{cases}
\end{equation}
where
\begin{align*}
A_0 &\doteq \Set{(t,\omega,\gamma(\cdot),j)\in\bar{A}\mid
- j\frac{\partial{q^0_t}}{\partial{u}} \le
\int_0^{+\infty} \bigl(j+\gamma(z)\bigr) e^{-\eta e^{R(T-t)}z}\frac{\partial{g(z,0)}}{\partial{u}}\pi_{t^-}(\lambda F_Z(dz))},\\
A_I &\doteq \Set{(t,\omega,\gamma(\cdot),j)\in\bar{A}\mid
- j\frac{\partial{q^I_t}}{\partial{u}} \ge
\int_0^{+\infty} \bigl(j+\gamma(z)\bigr) \frac{\partial{g(z,I)}}{\partial{u}}\pi_{t^-}(\lambda F_Z(dz))},
\end{align*}
and $0<\hat{u}(t,\omega,\gamma(\cdot),j)<I$ is the solution to
\begin{equation}
\label{eqn:null_deriv}
- j\frac{\partial{q^u_t}}{\partial{u}}
= \int_0^{+\infty} \bigl(j+\gamma(z)\bigr) e^{-\eta e^{R(T-t)}(z-g(z,u))}\frac{\partial{g(z,u)}}{\partial{u}}\pi_{t^-}(\lambda F_Z(dz)),
\end{equation}
for any $(t,\omega,\gamma(\cdot),j)\in (A_0\cup A_I)^C$. Then $u^*(t,\omega,\gamma(\cdot),j)$ is the unique maximizer of $f$, that is Eq. \eqref{eqn:u_maximizer} is valid.
\end{prop}
\begin{proof}
Since $f$ is continuous on the compact set $[0,I]$, it admits a maximum. Moreover, it is concave and the uniqueness of the maximizer is guaranteed. Now let us evaluate the first derivative of $f$:
\begin{multline}
\label{eqn:f_deriv1}
\frac{\partial{f(t,\omega,\gamma(\cdot),j,u)}}{\partial{u}}
= - j\eta e^{R(T-t)}\frac{\partial{q^u_t}}{\partial{u}} \\
-\int_0^{+\infty} \bigl(j+\gamma(z)\bigr)\eta e^{R(T-t)}e^{-\eta e^{R(T-t)}(z-g(z,u))}\frac{\partial{g(z,u)}}{\partial{u}}\pi_{t^-}(\lambda F_Z(dz)).
\end{multline}
Since
\begin{align*}
A_0 &= \Set{(t,\omega,\gamma(\cdot),j)\in\bar{A}\mid
\frac{\partial{f(t,\omega,\gamma(\cdot),j,0)}}{\partial{u}}\le0},\\
A_I &= \Set{(t,\omega,\gamma(\cdot),j)\in\bar{A}\mid
\frac{\partial{f(t,\omega,\gamma(\cdot),j,I)}}{\partial{u}}\ge0},
\end{align*}
by definition (see Eq. \eqref{eqn:f_deriv1}), using the concavity of $f$ we have that $\frac{\partial{f}}{\partial{u}}$ is decreasing in $u\in[0,I]$, hence $A_0\cap A_1=\emptyset$. Now there are only three possible cases. If $(t,\omega,\gamma(\cdot),j)\in A_0$, $f$ is decreasing in $u\in[0,I]$ and the maximizer is $u=0$. Similarly, if $(t,\omega,\gamma(\cdot),j)\in A_I$, $f$ is increasing in $u\in[0,I]$ and the maximizer is $u=I$. Finally, if $(t,\omega,\gamma(\cdot),j)\in (A_0\cup A_I)^C$, the maximizer coincides with the unique stationary point $\hat{u}(t,\omega,\gamma(\cdot),j)\in(0,I)$, that is the solution to Eq. \eqref{eqn:null_deriv}.
\end{proof}

%\begin{multline}
%\label{eqn:f_deriv2}
%\frac{\partial^2{f(t,\omega,\gamma,j,u)}}{\partial{u^2}}
%= - j\eta e^{R(T-t)}\frac{\partial^2{q^u_t}}{\partial{u^2}} \\
%-\int_0^{+\infty} \bigl(j+\gamma\bigr)\eta^2 e^{2R(T-t)}e^{-\eta e^{R(T-t)}(z-g(z,u))}\biggl(\frac{\partial{g(z,u)}}{\partial{u}}\biggr)^2\pi_{t^-}(\lambda F_Z(dz)) \\
%-\int_0^{+\infty} \bigl(j+\gamma\bigr)\eta e^{R(T-t)}e^{-\eta e^{R(T-t)}(z-g(z,u))}\frac{\partial^2{g(z,u)}}{\partial{u^2}}\pi_{t^-}(\lambda F_Z(dz)).
%\end{multline}

\begin{corollary}\label{cor_op}
Assume $g(z,u)$ differentiable in $u\in[0,I]$. Let $f$ be defined by Eq. \eqref{eqn:function_f} and suppose that it is strictly concave in $u$. 
Suppose that Assumption \ref{ass:verification} is fulfilled and let $(\hat{J_t},\hat{\Gamma}(t,z))\in\mathcal{L}^2\times\mathcal{\tilde{L}}^2$ be a solution to the BSDE \eqref{eqn:BSDE}. Let us define the control $\{u^*_t\doteq u^*(t,\omega,\hat{\Gamma}(t,z),\hat{J}_{t^-})\}_{t\in[0,T]}$, with the function $u^*(t,\omega,\gamma,j)$ given in Eq. \eqref{eqn:optimal_u}, that is
\begin{equation}
\label{eqn:optimal_u_process}
u^*(t,\omega,\hat{\Gamma}(t,z),\hat{J}_{t^-})=
\begin{cases}
	0 & \text{$(t,\omega)\in \tilde{A}_0$}
	\\
	\hat{u}(t,\omega,\hat{\Gamma}(t,z),\hat{J}_{t^-}) & \text{$(t,\omega)\in (\tilde{A}_0\cup \tilde{A}_I)^C$}
	\\
	I & \text{$(t,\omega)\in \tilde{A}_I$,}
\end{cases}
\end{equation}
where
\begin{align*}
\tilde{A}_0 &\doteq \Set{(t,\omega)\in[0,T]\times\Omega\mid
- \hat{J}_{t^-}\frac{\partial{q^0_t}}{\partial{u}} \le
\int_0^{+\infty} \bigl(\hat{J}_{t^-}+\hat{\Gamma}(t,z)\bigr) e^{-\eta e^{R(T-t)}z}\frac{\partial{g(z,0)}}{\partial{u}}\pi_{t^-}(\lambda F_Z(dz))},\\
\tilde{A}_I &\doteq \Set{(t,\omega)\in[0,T]\times\Omega\mid
- \hat{J}_{t^-}\frac{\partial{q^I_t}}{\partial{u}} \ge
\int_0^{+\infty} \bigl(\hat{J}_{t^-}+\hat{\Gamma}(t,z)\bigr) \frac{\partial{g(z,I)}}{\partial{u}}\pi_{t^-}(\lambda F_Z(dz))},
\end{align*}
and $0<\hat{u}(t,\omega,\Gamma(t,z),J_{t^-})<I$ is the solution to
\begin{equation}
\label{eqn:null_deriv_process}
- \hat{J}_{t^-}\frac{\partial{q^u_t}}{\partial{u}}
= \int_0^{+\infty} \bigl(\hat{J}_{t^-}+\hat{\Gamma}(t,z)\bigr) e^{-\eta e^{R(T-t)}(z-g(z,u))}\frac{\partial{g(z,u)}}{\partial{u}}\pi_{t^-}(\lambda F_Z(dz)).
\end{equation}
Then $\{u^*_t\}_{t\in[0,T]}$ is an optimal control.
\end{corollary}
\begin{proof}
By Proposition \ref{prop:admissibleproc} $u^*\in\mathcal{U}$. Since Eq. \eqref{eqn:u_maximizer} holds by Proposition \ref{prop:optimal_u}, then $u^*$ is an optimal control.
\end{proof}

Here we provide sufficient conditions for the concavity of $f$, which is the main hypothesis of Proposition \ref{prop:optimal_u}.

\begin{prop}\label{Con}
Suppose that the reinsurance premium $q^u_t$ and the self-insurance function $g(z,u)$ are linear or convex in $u\in[0,I]$. Then the function $f$ given in Eq. \eqref{eqn:function_f} is strictly concave in $u$.
\end{prop}
\begin{proof}
It follows directly by Eq. \eqref{eqn:function_f}.
\end{proof}

%Since $f$ given in Eq. \eqref{eqn:function_f} is the sum of two functions, it is sufficient to prove that they are concave separately. The first term $- j\eta e^{R(T-t)}q^u_t$ is clearly concave in $u\in[0,I]$, because $q^u_t$ is convex by hypothesis. Now the convexity of $g(z,u)$ in $u\in[0,I]$ implies the concavity of $z-g(z,u)$. This latter term is also non increasing in $u\in[0,I]$. Moreover, the negative exponential is convex and as a consequence the composite function $e^{-\eta e^{R(T-t)}(z-g(z,u))}$ turns out to be convex, which implies the thesis.
%This implies that $1-e^{-\eta e^{R(T-t)}(z-g(z,u))}$ is concave and hence the second term of Eq. \eqref{eqn:function_f} is concave. Now the proof is complete.\end{proof}

The following remark stress that the two hypotheses of the previous proposition are not merely technical conditions.

\begin{remark}
Both the classical premium calculation principles \eqref{eqn:evp} and \eqref{eqn:vp} and the proportional as well as the excess-of-loss reinsurance agreements satisfy the  hypotheses of Proposition \ref{Con}. In the next section we provide the explicit form of the optimal strategy in some special cases.
\end{remark}

\begin{remark}\label{diff}
When, $\forall (t,y) \in [0,T] \times \mathbb{R}$, the  distribution $F_Z(t,y,dz)$ admits a density function $f_Z(t,y,z)$, the differentiability of $g$ in Proposition \ref{prop:optimal_u} can be weakened by the hypothesis of $g$ differentiable in $u\in[0,I]$ for almost every $z\in[0,+\infty)$.
\end{remark}

%\begin{remark}By Example \ref{example:contracts}, we observe that the self-insurance function $g(z,u)$ is convex in $u\in[0,I]$ in the proportional as well as in the excess-of-loss reinsurance agreements. Hence in these popular cases the convexity of the reinsurance premium is sufficient to guarantee existence and uniqueness of an optimal strategy, confirming some existing results in the literature (see \cite[Lemma 4.1]{BC:IME2019} and \cite[Proposition 7]{BCrisks}).\end{remark}

%--------------------------------------------------------------------------------
%	PROPERTIES OF U^*
%--------------------------------------------------------------------------------

\section{Some properties of the optimal reinsurance strategy}
\label{section:comparison}

In this Section we investigate some properties of the optimal reinsurance strategy. In Subsection \ref{subsec:5.1} we prove that if the premia satisfy the Markovian property in the filter process, then the same property applies to the optimal strategy. This means that the optimal strategy depends on the estimate of  the  environmental stochastic factor distribution given the available information. Next, in Subsection \ref{subsec:5.2} we perform a sensitivity analysis and in Subsection \ref{subsec:5.3} we give a comparison result with the full information case for some relevant examples.  In particular, we extend the comparison made in \cite{liangbayraktar:optreins} for the Markov modulated risk model under the proportional reinsurance to the case of $Y$ having infinitely many states and to the excess of loss reinsurance contract.\\

\subsection{Markovianity in the filter process}
\label{subsec:5.1}
Assuming premia at time $t$ depending on the filter process $\pi_{t^-}$, as in the classical premium calculation principles, see Example \ref{example:premia}. Let  $ {\cal P}(\mathbb{R})$ be the space of probability measures on $\mathbb{R}$ endowed with the weak topology. 
Let us observe  that $\{\pi_t\}_{t \in [0,T]}$ is an $\mathbb{F}$-Markov process taking values in ${\cal P}(\mathbb{R})$ (see Eq. \eqref{eqn:KSeq_appendix}).
Then the value process $\{V_t\}_{t \in [0,T]}$, given in \eqref{eqn:V_t}, is such that $V_t = v(t, \pi_{t^-})$, with $v(t, \pi)$ measurable function on the space $[0,T] \times {\cal P}(\mathbb{R})$.
  
Let us observe that by Eq. \eqref{eqn:J^u} we have that
\[
\begin{split}
J^I_{T_n} - J^I_{T_n^-} &=  \Gamma(T_n,Z_n) \\
& = e^{-\eta \bar{X}^I_{T_n} e^{R T}} V_{T_n} - e^{-\eta \bar{X}^I_{T^-_n} e^{R T}} V_{T^-_n}  = 
e^{-\eta \bar{X}^I_{T^-_n}e^{R T}} ( V_{T_n} e^{\eta Z_n e^{R (T- T_n)}}  -  V_{T^-_n} ) .
\end{split}
\]

Denote by $W(t,\pi, z) : [0,T] \times {\cal P}(\mathbb{R}) \times [0, + \infty)  \to {\cal P}(\mathbb{R})$ a the measurable function such that 
\eqref{jump} is fulfilled,  that is $\pi_{T_n}(f) = W(T_n,\pi_{T_n^-}, Z_n)(f)$, $\forall f \in \mathcal{D}^Y$.
%$$$W(t,\pi_{t^-}, z)(f) = w_t^\pi(f,z) + \pi_{t^-}(f),  \quad \forall f \in \mathcal{D}^Y$$where $w^\pi_{t^-}(f,z)$ is given in Eq. \eqref{eqn:w^pi}. %(this means that $\pi_{T_n}(f) = W(T_n,\pi_{T_n^-}, Z_n)(f)$). 
Then we have that $V_{T_n} = v(T_n,  W(T_n, \pi_{T_n^-}, Z_n))$,
so by Remark \ref{remark:gamma_unique} we can write
 $$\Gamma(t,z) = e^{-\eta \bar{X}^I_{t^-} e^{RT}} \big ( v(t,  W(t,\pi_{t^-}, z)) e^{\eta z e^{R (T-t)}}  - v(t, \pi_{t^-}) \big) \quad  \pi_{t^-}(\lambda F_Z(dz))(\omega) \mathbb{P}(d\omega)\,dt\text{-a.e.}.$$
and, as a consequence,
\begin{equation}
\label{nuova}
\big (J^I_{t^-} + \Gamma(t,z) \big ) e^{-\eta z e^{R (T- t)}}  = v(t, W(t,\pi_{t^-}, z)) e^{-\eta \bar{X}^I_{t^-} e^{R T}}   \quad \pi_{t^-}(\lambda F_Z(dz))(\omega) \mathbb{P}(d\omega)\,dt\text{-a.e.}.
\end{equation}

We are now able to prove that the reinsurance optimal strategy is a filter-feedback control, this means that at time $t$ only depends on the estimate of  the distribution of the environmental stochastic factor  immediately before time $t$.

\begin{prop}\label{M1} 
Assume that the premia $c_t$ and $q^u_t$, $\forall u \in [0,I]$, at time $t$ depend on the filter process $\pi_{t^-}$. 
Then  the optimal reinsurance strategy  is Markovian  in the filter process, that is $u^*_t = u^*(t, \pi_{t^-})$, with
$u^*(t, \pi)$ being a measurable function of $(t, \pi) \in [0,T] \times  {\cal P}(\mathbb{R})$.
\end{prop}

\begin{proof}
Recall that the optimal reinsurance strategy is the maximizer of $f(t,\Gamma(t,z), J^I_{t^-},u)$ (given in Eq. \eqref{eqn:driver_f}) over the class of admissible controls. 
By 
\begin{equation}
\label{eqn:Jvsec5}
J^I_{t^-} =   e^{-\eta \bar{X}^I_{t^-} e^{R T} } v(t, \pi_{t^-}),
\end{equation}
and Eq. \eqref{nuova}, one gets that %\begin{align}\label{C1}f(t,\Gamma(t,z), J^I_{t^-},u) &=- J^I_{t^-} \eta e^{R(T-t)}q^u_t\notag\\&\ + \int_0^{+\infty} \ ( J^I_{t^-}+\Gamma(t,z)) e^{-\eta e^{R(T-t)}z}  \big (e^{\eta e^{R(T-t)}z} - e^{\eta e^{R(T-t)} g(z,u)}  \big) \pi_{t^-}(\lambda F_Z(dz)) ,\end{align}
\begin{equation}\label{m1}
f(t,\Gamma(t,z), J^I_{t^-},u) = e^{-\eta \bar{X}^I_{t^-} e^{R T} } h(t, \pi_{t^-}, u) \qquad \forall u \in [0,I],
\end{equation}
where
\begin{equation}\label{m2}
\begin{split}
h(t, \pi_{t^-}, u) &\doteq -  v(t, \pi_{t^-}) \eta e^{R(T-t)} q^u_t  \\
&\ + \int_0^{+\infty} v(t,  W(t,\pi_{t^-}, z)) (e^{\eta e^{R(T-t)}z} 
- e^{\eta e^{R(T-t)} g(z,u)} )\pi_{t^-}(\lambda F_Z(dz)) .
\end{split}
\end{equation}

Hence our result follows by measurability selection theorems.
\end{proof}

\begin{remark}
Notice that, assuming premia $c_t$ and $q^u_t$, $\forall u \in [0,I]$, at time $t$ depending of the filter process $\pi_{t^-}$, 
the pair $\{(\bar{X}^I_{t}, \pi_{t})\}_{t \in [0,T]}$ is an $\mathbb{F}$-Markov process and  
$J^I_t = \tilde v (t, \bar{X}^I_{t}, \pi_{t^-})$ with $\tilde v(t,x, \pi) \doteq e^{-\eta x e^{RT} } v(t, \pi)$. 
In the Markov modulated risk model, that is when $Y$ is  a continuous time Markov chain taking values in 
${\cal S} = \{1, \dots, M\}$,  the pair $\{(\bar{X}^I_{t}, \pi_{t})\}_{t \in [0,T]}$ is an  $(M+1)$-dimensional  $\mathbb{F}$-Markov process 
and $\tilde v (t, x, \pi)$ can be characterized in terms of the associated HJB-equation, if it is regular enough.
Concerning that point, see \cite{liangbayraktar:optreins}, where the problem is discussed in a Markov modulated risk model and the authors make use of a generalized HJB equation, introducing a weaker notion of differentiability.  
In the general case this approach is not suitable since the filter is an infinite-dimensional process and this motivates the characterization in terms of BSDEs, as proposed in this paper. \end{remark}

\subsection{Effect of the safety loading}
\label{subsec:5.2}

In this subsection we determine the effect of the reinsurance safety loading $\theta>0$ on the optimal reinsurance strategy in the case of proportional  contract,
that is $g(z,u) = zu$, $u \in [0,1]$ (see Example \ref{example:contracts}) and under the expected value principle (see Example \ref{example:premia}). We will show that the greater is the value of $\theta$, which implies a greater reinsurance premium, the greater will be the optimal retention level. This is consistent with the classical law of demand in economics and with existing results. Let
\begin{equation}\label{B}
B(t, \pi) \doteq  \frac{ \int_0^{+\infty} \frac{v(t,  W(t,\pi, z))}{v(t, \pi)}  z \pi(\lambda F_Z(dz)) } 
{ \int_0^{+\infty} z \pi(\lambda F_Z(dz)) } -1 ,
\end{equation}

\begin{equation}\label{D}
D(t, \pi) \doteq \frac{\int_0^{+\infty} \frac{v(t,  W(t,\pi, z))}{v(t, \pi)} e^{\eta e^{R(T-t) z} }z \pi (\lambda F_Z(dz))} 
{\int_0^{+\infty} z \pi(\lambda F_Z(dz))} -1 .
\end{equation}

\begin{prop}\label{p1}
In the proportional reinsurance, under the expected value principle the optimal reinsurance strategy increases with respect the reinsurance safety loading $\theta$.
Furthermore, it is given by

\begin{equation}\label{safety}
u^*(t,\pi_{t^-})=
\begin{cases}
	0 & \text{$0 < \theta \leq B(t, \pi_{t^-}) $}
	\\
	\hat{u}(t,\pi_{t^-}) & \text{$B(t, \pi_{t^-}) < \theta \leq D(t, \pi_{t^-}) $}
	\\
	1 & \text{$\theta\geq D(t, \pi_{t^-}) $,}
\end{cases}
\end{equation}
where $B(t, \pi)$ and $D(t, \pi)$ are defined in \eqref{B} and \eqref{D}, respectively, and $0<\hat{u}(t,\pi_{t^-})<I$ is the unique solution to  

\begin{equation}\label{eqq}
(1 + \theta) \int_0^{+ \infty} z  \pi_{t^-}(\lambda F_Z(dz))  = \int_0^{+\infty}  \frac{v(t,  W(t,\pi_{t^-}, z))}{v(t, \pi_{t^-})} e^{\eta e^{R(T-t)}zu} z \pi_{t^-}(\lambda F_Z(dz)) \doteq G(u). 
\end{equation}
\end{prop}
\begin{proof}
Under the expected value principle, see Eq. \eqref{eqn:evp}, we have that
 \begin{equation}
 \begin{split}
 h(t, \pi, u) &= -  v(t, \pi) \eta e^{R(T-t)} (1 + \theta) \int_0^{+ \infty} z (1  - u) \pi(\lambda F_Z(dz))\\
&\ + \int_0^{+\infty} v(t,  W(t,\pi, z)) (e^{\eta e^{R(T-t)}z} 
- e^{\eta e^{R(T-t)} zu} )\pi(\lambda F_Z(dz))
\end{split}
\end{equation}
is strictly  concave in $u\in[0,I]$ and, taking into account Eq. \eqref{m1}, $f$ is so. Using Eq. \eqref{nuova} and Eq. \eqref{eqn:Jvsec5} we notice that Eq. \eqref{eqn:null_deriv_process} can be rewritten as \eqref{eqq}. The right hand term in this equation is an increasing function on $u \in[0,1]$, therefore $\hat{u}(t,\pi)$ increases with respect to $\theta$ (see Figure \ref{img:prop52}).

Finally, using Corollary \ref{cor_op} we get the explicit form of the optimal strategy and the result readily follows. 
\end{proof}

\begin{figure}[ht]
\centering
\includegraphics[scale=0.4]{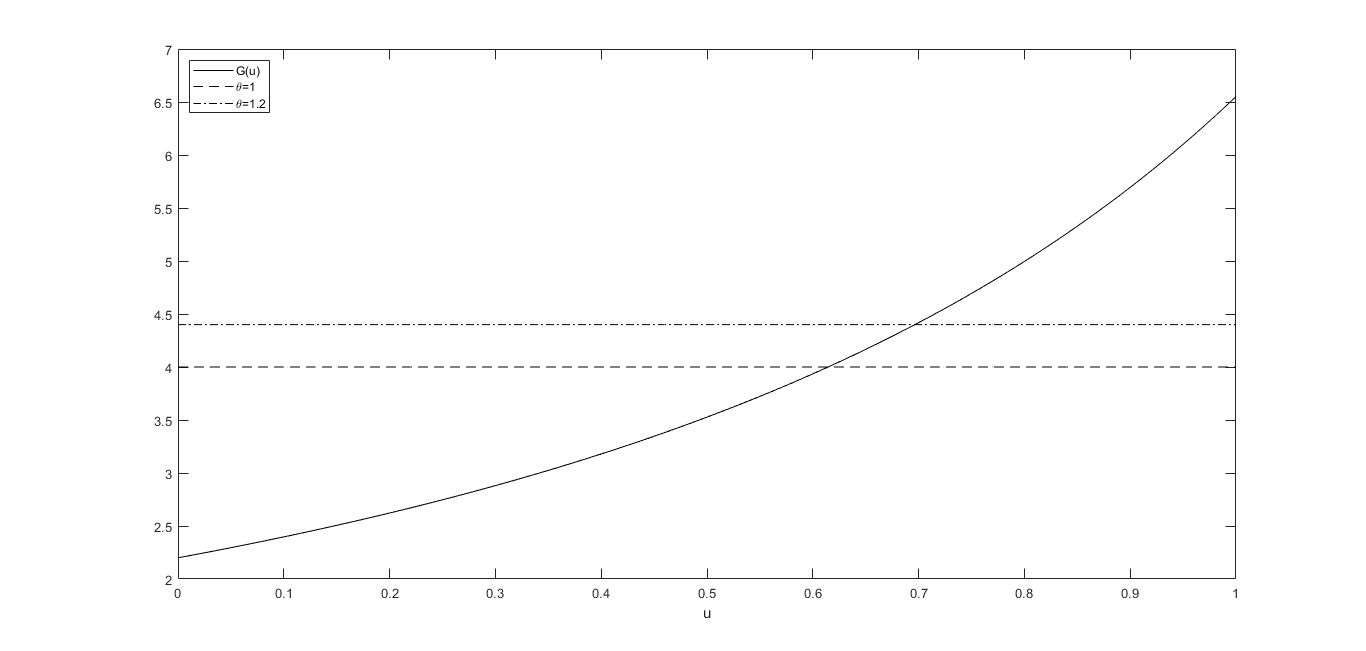}
\caption{The effect of $\theta$ on the reinsurance strategy.}
\label{img:prop52}
\end{figure}

\subsection{Comparison with the case of complete information}
\label{subsec:5.3}

In this subsection we compare the optimal strategy under partial information to the one with full information.  
In some special cases (see Propositions \ref{P1} and \ref{P2} below) we can prove that the optimal retention level under partial information is smaller than the one in the full information case. This means that the insurer who takes into account a partial information framework tends to buy an additional protection with respect to the (theoretical) case of complete information.

We consider the case of unknown time-homogeneous jump intensity (i.e. $\lambda(t,y) = \lambda(y)$) and known claims size distribution (i.e. $F_Z(t,y,dz) = F_Z(dz)$).
Moreover, we suppose that the stochastic factor $Y$ takes value in a discrete set ${\cal S} = \{1, 2, \dots \}$. Let us recall (see Remarks \ref{remark:filterapp} and \ref{remark:MMRM}) that in this case the filter is described by the sequence
$\pi_{t}(i) = P(Y_t=i | {\cal F}_t), i \in {\cal S}$ and $W(t,\pi_{t^-}, z) = W(\pi_{t^-}) =   \{W_i(\pi_{t^-}),  i \in {\cal S}\}$ by Eq. \eqref{eqn:filterWlast}, where
\begin{equation}
W_i(\pi_{t^-}) = { \lambda(i) \pi_{t^-}(i)   \over  \sum_{j \in  {\cal S}} \lambda(j) \pi_{t^-}(j)}, \quad i \in {\cal S}.
\end{equation}

Without loss of generality we assume 
$\lambda(1) \leq \lambda(2) \leq \dots$ and following the same lines as in \cite[Lemma 4.1]{liangbayraktar:optreins} (where the case with a finite set ${\cal S}$ is discussed), see also \cite[Theorem 5.6]{bauerle2007}, we can prove that%
\footnote{The result essentially follows from the stochastic dominance of Poisson processes with increasing intensities and by this inequality: $\sum_{j\in\cal{S}}\pi_j\alpha_j\beta_j\le\sum_{j\in\cal{S}}\pi_j\alpha_j\sum_{j\in\cal{S}}\pi_j\beta_j$, where $\{\alpha_j\}_{j\in\cal{S}}$ is an increasing sequence and $\{\beta_j\}_{j\in\cal{S}}$ is increasing, while the nonnegative sequence $\{\pi_j\}_{j\in\cal{S}}$ is such that $\sum_{j\in\cal{S}}\pi_j=1$.}
\begin{equation}
\label{dis}
v(t,  W(\pi) ) \geq v(t, \pi), \quad  \forall t \in  [0,T], \pi \in \{ \{\pi_i\}_{i \in {\cal S}}; \quad \sum_{i \in {\cal S}} \pi_i  =1, \pi_i \in  [0,1]\}
\end{equation}
with $W(\pi) = \{W_i(\pi),  i \in {\cal S}\}$.

%\nuovo{Non so se osservare che affinch\'e tutto funzioni (serie convergenti) la funzione valore deve essere finita e $\sum_{j\in\cal{S}}\lambda(j)<+\infty$.}\nuovo{PER NOI: consideriamo $\bar{v}$ tale che (definita nel modo ovvio):
%\[v(t,\pi) = \inf_{u\in\mathcal{U}_t}\bar{v}(t,\pi,u) .\]
%Allora abbiamo che per ogni $u\in\mathcal{U}_t$\[\bar{v}(t,\pi,u) = \sum_{i \in {\cal S}} \pi_i \bar{v}(t,e_i,u) ,\]
%e, analogamente,\[\bar{v}(t,\pi,u) = \sum_{i \in {\cal S}} { \lambda(i) \pi_{t^-}(i)   \over  \sum_{j \in  {\cal S}} \lambda(j) \pi_{t^-}(j)} \bar{v}(t,e_i,u) ,\]Osserviamo che $\bar{v}(t,e_i,u)\ge\bar{v}(t,e_{i+1},u)$ $\forall i\ge1$ (per la dominanza) e le intensit\`a sono ordinate per ipotesi. A questo punto applichiamo la disuguaglianza per le serie.}

\begin{prop}\label{P1}
Let the assumptions of this subsection be satisfied. Under the proportional reinsurance and the premium calculation principles in Example \ref{example:premia},
the optimal reinsurance strategy under partial information is always less or equal to the one under full information.
\end{prop}

\begin{proof}
We analyze two premium calculation principles and, correspondingly, we divide the proof in two parts.

\textit{Expected value principle}

Under the expected value principle (see Eq. \eqref{eqn:evp}) and a proportional reinsurance (i.e. $g(z,u)=uz, u\in[0,1]$ by Example \ref{example:contracts}), using Proposition \ref{prop:optimal_u} and Corollary \ref{cor_op} we easily obtain that the optimal reinsurance strategy is given by 
\begin{equation}
\label{eqn:optimal_u_nuova}
u^*_t =
\begin{cases}
	0 & \text{if $1+ \theta \leq \frac{v(t,  W(\pi_{t^-}))}{v(t, \pi_{t^-})} $}
	\\
	1 & \text{if $(1+ \theta) \mathbb{E}[Z] \geq  \int_0^{+\infty} \frac{v(t,  W(\pi_{t^-}))}{v(t, \pi_{t^-})} z e^{\eta e^{R(T-t)}z}  F_Z(dz)$ }
	\\
	\hat{u}(t,  \pi_{t^-})& \text{ otherwise,}
\end{cases}
\end{equation}
where $\hat{u}$ is the unique solution to 
 \begin{equation}
 \label{null_deriv_nuova1}   
(1 + \theta) \mathbb{E}[Z]  
= \frac{v(t,  W(\pi_{t^-}))}{v(t, \pi_{t^-})} \int_0^{+\infty}  ze^{\eta e^{R(T-t)}uz}  F_Z(dz) \doteq h_1(t, \pi_{t^-},u).
\end{equation}
 
In the full information case, from \cite[Lemma 4.2]{BC:IME2019}, full reinsurance is never optimal, the optimal reinsurance strategy is a deterministic function of time and it is given by
\begin{equation}
\label{eqn:optimal_u_f}
u^{*,f}(t )=
\begin{cases}
	1 & \text{if $(1+ \theta) \mathbb{E}[Z] \geq  \int_0^{+\infty} z e^{\eta e^{R(T-t)}z}  F_Z(dz) $}
	\\
	\hat{u}^f(t) & \text{otherwise,}
\end{cases}
\end{equation}
where $\hat{u}^f$ is the unique solution to  
\begin{equation}
 \label{null_deriv_full}   
(1 + \theta) \mathbb{E}[Z]
= \int_0^{+\infty}  z e^{\eta e^{R(T-t)}uz}  F_Z(dz) \doteq h_0(t,u)
\end{equation}

Let us consider the equations \eqref{null_deriv_nuova1} and \eqref{null_deriv_full} defined for all $u\in\mathbb{R}$, then equations \eqref{eqn:optimal_u_nuova} and  \eqref{eqn:optimal_u_f} can be written as
\[
u^*_t =0\land\hat{u}_t\lor1 \qquad\text{and}\qquad u^{*,f}(t )=\hat{u}^f(t)\lor1 ,
\]
respectively. Since
\[
h_1(t, \pi_{t^-},u) = \frac{v(t,  W(\pi_{t^-}))}{v(t, \pi_{t^-})}h_0(t,u)\ge h_0(t,u) \qquad \forall u\in[0,1],
\]
and both sides are increasing in $u$, in order to have
\[
h_1(t, \pi_{t^-},\hat{u}_t) = 1+\theta = h_0(t,\hat{u}^f(t)),
\]
we must have that $\hat{u}_t\le \hat{u}^f(t)$ (see Figure \ref{img:prop53}), which implies our statement.

%\nuovo{Questo si pu\`o cancellare se va bene come scritto sopra:\\By \eqref{dis} we have that null reinsurance in the partial information case, that is $u^*_t =1$, implies null reinsurance in the full informationcase, $u^{*,f}(t )=1$ and 
%Denote $u^*_t= u^*(t,  \pi_{t^-})$ and $u^{*,f}(t)$ optimal reinsurance startegies under partial and full information, respectively.since $h_1(t, \pi,u) \geq h_0(t,u)$ we obtain the inequality $u^*_t \leq u^{*,f}(t)$.\\Moreover let us observe that if the safety loading $\theta$ is such that$$(1+ \theta) E(Z) \leq  \int_0^{+\infty} z e^{\eta e^{R(T-t)}z}  F_Z(dz)$$ null reinsurance is never optimal in the full information case and this implies that also under the partial information it is, that is  $u^{*,f}(t )= \hat{u}^f(t)$and  
%\begin{equation}
%\label{eqn:optimal_u_bis}
%u^*_t =
%\begin{cases}
%	0 & \text{ if $1+ \theta \leq \frac{v(t,  W(\pi_{t^-}))}{v(t, \pi_{t^-})} $}
%	\\
%	\hat{u}_t & \text{ otherwise.}\end{cases}\end{equation}}

\begin{figure}
\centering
\includegraphics[scale=0.4]{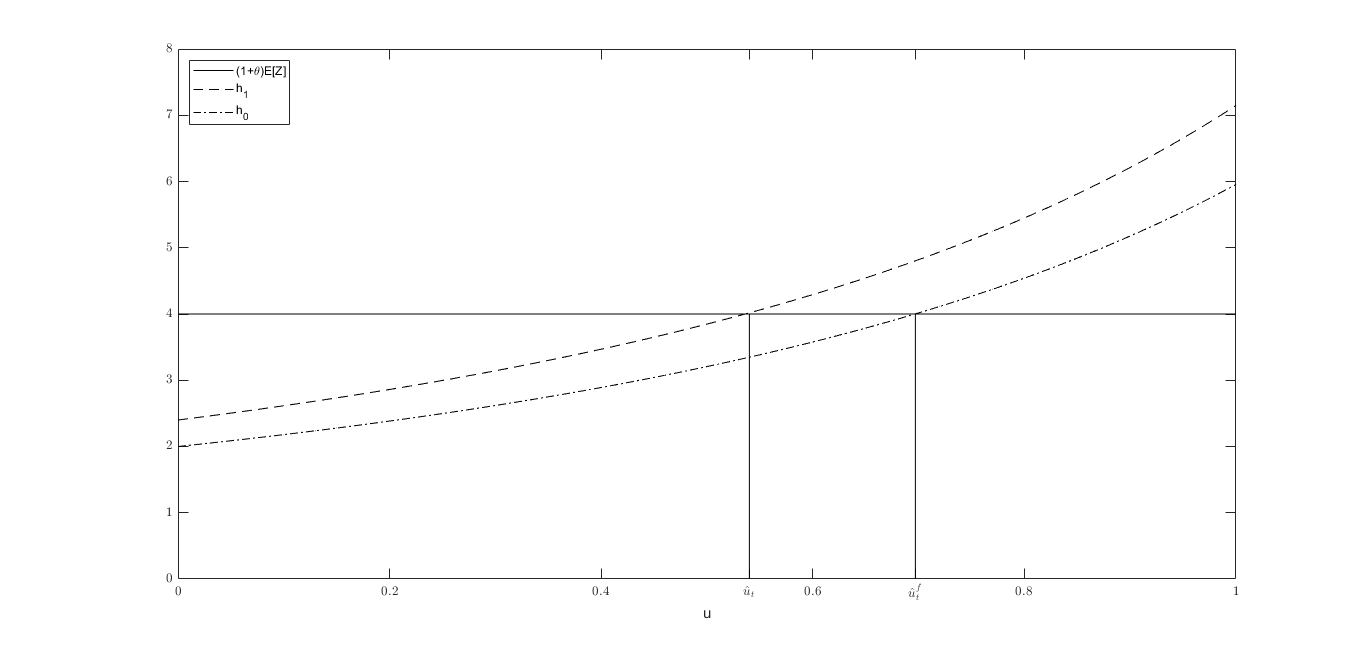}
\caption{Functions $h_1$ and $h_0$ giving $\hat{u}_t$ and $\hat{u}^f(t)$ under the expected value principle.}
\label{img:prop53}
\end{figure}

\textit{Variance premium principle}

Now let us denote 
\[
H(u) \doteq  \mathbb{E}[Z]  +  2 \theta (1-u) \mathbb{E}[Z^2] .
\]
By Corollary \ref{cor_op} we  obtain that under the variance premium principle the optimal reinsurance strategy is given by 
\begin{equation}
\label{eqn:optimal_u_nuova2}
u^*_t =
\begin{cases}
	0 & \text{if $ 2 \theta \mathbb{E}[Z^2]  \leq  (\frac{v(t,  W(\pi_{t^-}))}{v(t, \pi_{t^-})} -1) \mathbb{E}[Z] $}
	\\
	%1 & \text{if $(1+ \theta) \mathbb{E}[Z] \geq  \int_0^{+\infty} \frac{v(t,  W(\pi_{t^-}))}{v(t, \pi_{t^-})} z e^{\eta e^{R(T-t)}z}  F_Z(dz)$ }\\
	\hat{u}(t,  \pi_{t^-})& \text{ otherwise,}
\end{cases}
\end{equation}
where $\hat{u}$ is the unique solution to 
 \begin{equation}
 \label{null_deriv_nuova2}   
H(u)  = h_1(t, \pi_{t^-},u).
\end{equation}
%Under the  variance principle (see Eq. \eqref{eqn:vp}) and proportional reinsurance, Eq. \eqref{null_deriv_nuova}  reads as
%$$\mathbb{E}[Z] + 2 \theta (1-u) \mathbb{E}[Z^2] = h_1(t, \pi_{t^-},u), $$
In the full information case (see \cite[Lemma 4.3]{BC:IME2019}) the optimal strategy is a deterministic function  of time and is given by $u^{*,f}(t )= \hat{u}^f(t)$, where  $\hat{u}^f$ is the unique solution to 

$$H(u) = h_0(t,u).$$
Hence the inequality $u^*_t \leq u^{*,f}(t)$ immediately follows by the same arguments as in the expected value principle case (see Figure \ref{img:prop53vp}).
\begin{figure}
\centering
\includegraphics[scale=0.4]{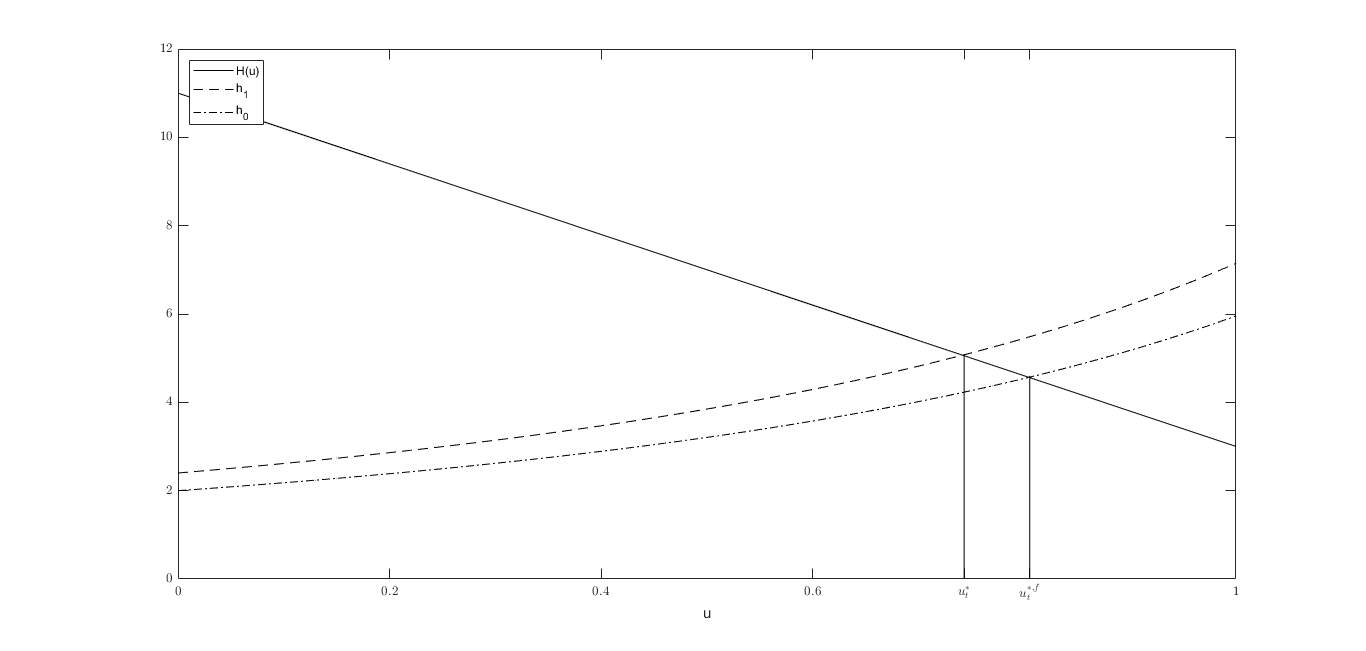}
\caption{Functions $h_1$ and $h_0$ giving $u^*_t$ and $u^{*,f}(t)$ under the variance premium principle.}
\label{img:prop53vp}
\end{figure}
\end{proof}

\begin{prop}\label{P2}
Let the assumptions of this subsection be satisfied and $F_Z(dz)$ admit density.  Under the excess of loss reinsurance and the expected value principle 
the optimal reinsurance strategy is always less or equal to the one under full information.
\end{prop}
\begin{proof}
Consider the excess-of-loss reinsurance (see Example \ref{example:contracts}) and the expected value principle (see Eq. \eqref{eqn:evp}). In the full information case, see \cite[Proposition 8]{BCrisks}, the optimal reinsurance strategy satisfies
\begin{equation*}
(1 + \theta)(1-  F_Z(u)) = e^{\eta e^{R(T-t)}g(z,u))}  (1-  F_Z(u)),
\end{equation*}
so that it is given by
\begin{equation*}
\label{eqn:excessfull_u*}
u^{*,f}(t)= \frac{1} {\eta} e^{-R(T-t)} \log(1 + \theta).
\end{equation*}
In the partial information framework, taking into account Remark \ref{diff}, Eq. \eqref{eqn:null_deriv_process} reads as
 \begin{equation*}
% \label{null_deriv_nuova3}
(1 + \theta) (1-  F_Z(u))
= \frac{v(t,  W(\pi_{t^-}))}{v(t, \pi_{t^-})}  e^{\eta e^{R(T-t)}u} (1-  F_Z(u)),
\end{equation*}
and %performing calculations similar to \cite[Proposition 8]{BCrisks} 
we find out that
$$u^*_t= \frac{1} {\eta} e^{-R(T-t)} \log\Big ( (1 + \theta) \frac{v(t, \pi_{t^-})}{v(t,  W(\pi_{t^-})} \Big ).$$

It is easy to see that $u^*_t \leq u^{*,f}(t)$ by the inequality \eqref{dis}. 
\end{proof}

\section{Conclusions}

This paper extends the existing results on optimal reinsurance in many directions. We introduce a general risk model where both the claims arrival intensity and the claim size distribution are affected by an environmental stochastic factor $Y$, which is modeled as a general Markov process (in \cite{liangbayraktar:optreins} it was a finite state Markov chain and in \cite{BCrisks} a real-valued diffusion process). This model formulation allows the insurer to take into account risk fluctuations. However, it is well known that the insurer has only a partial information at disposal. Namely, she only observes the claims arrival times and the corresponding amount. Hence $Y$ is supposed to be unobservable and  as a consequence claims arrival intensity and the claim size distribution has to be inferred from the observations. Considering general premium and reinsurance contract, we solve the optimization problem characterizing the value process and the optimal strategy in terms of a solution to a BSDE. Our results show that the insurer would react to risk fluctuations by modifying the reinsurance policy. 
Some examples for classical reinsurance agreements are illustrated. By analyzing the effect of safety loading on the optimal strategy we determine the price that the insurer deems reasonable for the reinsurer assuming part of her risks.  Finally, we show that insurer with partial information is more conservative with respect the insurer with complete information.

%--------------------------------------------------------------------------------
%	ACKNOWLEDGMENTS AND INTEREST
%--------------------------------------------------------------------------------

\section*{Acknowledgements}
The authors are partially supported by the GNAMPA Research Project 2019 (\textit{Problemi di controllo ottimo stocastico con osservazione parziale in dimensione infinita}) of INdAM (Istituto Nazionale  di Alta Matematica). We are also grateful to anonymous referees for their helpful comments.

\section*{Declaration of interest}
None.

%--------------------------------------------------------------------------------
%	APPENDIX
%--------------------------------------------------------------------------------

\appendix

\section{Filtering with marked point processes observations}

Here, we recall the main results on filtering with marked point processes observations. 
Under Assumption \ref{Y}, the filter can be characterized as the unique strong solution of the so called Kushner-Stratonovich equation. We refer to \cite{ceci:2006} and \cite{cc:AAP2012} for a detailed proof. 

\begin{theorem}[KS-equation]
Under Assumption \ref{Y}, the filter $\pi$ is the unique strong solution to the Kushner-Stratonovich equation,  
for any bounded function $f \in \mathcal{D}^Y$
\begin{equation}
\label{eqn:KSeq_appendix}
d \pi_{t}(f) = \pi_{t} (\mathcal{L}^Y f) dt + \int_0^{+ \infty} w^\pi_{t}(f,z)(m(dt,dz) - \pi_{t^-}(\lambda F_Z(dz)) dt), \quad \pi_0(f) = f(0,y_0),
\end{equation}
where
\begin{equation}
\label{eqn:w^pi}
w^\pi_{t}(f,z) \doteq \frac{d \pi_{t^-}(\lambda F_Z f)} {d \pi_{t^-}(\lambda F_Z ) }(z) - \pi_{t^-}(f) +
 \frac{d\pi_{t^-}({\bar{\cal L}} f)}{d\pi_{t^-}(\lambda F_Z ) }(z).
\end{equation}
 
Here $\mathcal{\bar{L}}$ is an operator which takes into account possible common jump times between $Y$ and $m(dt,dz)$, while ${d \pi_{t^-}(\lambda F_Z f) \over d \pi_{t^-}(\lambda F_Z ) }(z)$ and ${d\pi_{t^-}(\mathcal{\bar{L}}f) \over d\pi_{t^-}(\lambda F_Z ) }(z)$ denote the Radon-Nikodym derivatives of the measures $\pi_{t^-}(\lambda F_Z (dz)f)$ and $\pi_{t^-}(\mathcal{\bar{L}}f (dz))$ with respect to $ \pi_{t^-}(\lambda F_Z(dz))$, respectively.
\end{theorem}

The filtering equation has a natural recursive structure. In fact, between two consecutive jump times, $t \in (T_{n-1}, T_n)$, the equation reads as:
\begin{equation}
\label{eq_filter}
d \pi_{t}(f) = [ \pi_{t} (\mathcal{L}^0 f) +  \pi_{t}(f) \pi_{t}(\lambda)  - \pi_{t}(\lambda f) ] dt,
\end{equation}
where $\mathcal{L}^0 f \doteq \mathcal{L}^Y f - \mathcal{\bar{L}}f $ and coincides with $\mathcal{L}^Y$ if there are not common jump times between state and observations.

At a jump time $T_n$:
$$ \pi_{T_n}(f) = {d \pi_{T_n^-}(\lambda F_Z f) \over d \pi_{T_n^-}(\lambda F_Z ) }(Z_n) +
 {d\pi_{T_n^-}(\mathcal{\bar{L}}f) \over d\pi_{T_n^-}(\lambda F_Z ) }(Z_n).$$ 
 
Hence $\pi_{T_n}(f)$ is completely determined by the observed data $(T_n, Z_n)$ and by the knowledge of $\pi_{t}$ in the interval $t \in [T_{n-1}, T_n)$.  

Let us observe that between two consecutive jump times the filter solves a non-linear deterministic equation (see Eq. \eqref{eq_filter}). We are able to provide a computable solution by means of a linearized method (see \cite[Lemma 3.1]{cg:2006}). For simplicity, we assume no common jump times between $Y$ and $m(dt, dz)$ in the sequel.

\begin{prop}\label{ro}

Let $\rho^n$ a process with values in the set of positive finite measures on $\mathbb{R}$ solution to the linear equation 
$$d \rho^n_t(f) = \rho^n_t (\mathcal{L}^Y f - \lambda f) dt , \quad  \rho^n_{T_{n-1}}(f) = \pi_{T_{n-1}}(f), \quad  t \in (T_{n-1}, T_n).$$
Then the process 
$$ {\rho^n_t(f) \over  \rho^n_t(1)}, \qquad t \in (T_{n-1}, T_n),$$
solves  Eq. \eqref{eq_filter}.
Moreover the following representation holds
$$\rho^n_t(f) = E_{n-1}[ f(t, Y_t) e^{- \int_s^t \lambda(r, Y_r) dr} ] |_{s= T_{n-1}}, $$
where $E_{n-1}$ denotes the conditional  expectation given  the distribution $Y_{T_{n-1}} $ equal to $\pi_{T_{n-1}}$.
%Under Assumption \ref{Y} the following representation holds
%
%$$\pi_{t^-}(f) = { E_{n-1}[ f(t, Y_t) e^{- \int_s^t \lambda(r, Y_r) dr} ] |_{s= T_{n-1}}\over E_{n-1}[ e^{- \int_s^t \lambda(r, Y_r) dr}] |_{s= T_{n-1}}}, \quad t \in (T_{n-1}, T_n), $$
\end{prop}

Finally, Proposition \ref{new} is a direct consequence of Proposition \ref{ro} and of the strong uniqueness of solution to the Kushner-Stratonovich equation \eqref{eqn:KSeq_appendix}.\\
 
In the last part of the section we discuss same special cases.  

\begin{example}[Known jump size distribution and unknown intensity]
\label{example:A1}
Let $F_Z(t, y,dz)=F_Z(dz)$, then the filtering equation \eqref{eq_filter} reduces to
$$d \pi_{t}(f) = \pi_{t} (\mathcal{L}^Y f) dt + {\pi_{t^-}(\lambda f)  - \pi_{t^-}(f) \pi_{t^-}(\lambda) \over \pi_{t^-}(\lambda)} (dN_t - \pi_{t^-}(\lambda) dt),$$
where $N_t = m((0,t] \times [0,+\infty)) = \sum_{n\ge1} \mathbbm{1}_{\{T_n\le t\}}$ is the claims arrival process.
Between two consecutive jump times, $t \in (T_{n-1}, T_n)$:
$$d \pi_{t}(f) = [ \pi_{t} (\mathcal{L}^Y f)  - \pi_{t}(\lambda f)  + \pi_{t}(f) \pi_{t}(\lambda) ] dt,$$
while at a jump time $T_n$:
$$ \pi_{T_n}(f) = W(T_n,\pi_{T_n^-}) \doteq {\pi_{T_n^-}(\lambda f) \over  \pi_{T_n^-}(\lambda)},$$
which coincides with Eq. \eqref{eqn:filter1} in Remark \ref{remark:filterapp}.
\end{example}

\begin{example}[Markov Modulated Risk Model with infinitely many states]
\label{example:A2}
Now we consider the case where $Y$ is a continuous time Markov chain taking values in a discrete  set ${\cal S} = \{1,2, \dots\}$ and $\{a_{ij}\}_{i\in {\cal S}, j \in {\cal S}}$ its generator matrix. 
Here, $a_{ij}>0$, $i \neq j$, gives the intensity of a transition from state $i$ to state $j$, and it is such that $\sum_{j\geq 1, j \neq i}a_{ij} = - a_{ii}$.
Defining the functions $f_i(y):=\mathbbm{1}_{y= i}$, $i \in {\cal S}$, the filter is completely described via the knowledge of 
$\pi_{t}(i):= \pi_{t}(f_i) = P(Y_t = i \mid\mathcal{F}_t)$, $i \in {\cal S}$, because for every function $f$ we have that 
$$\pi_{t}(f) = \sum_{i \in  {\cal S}} f(t,i) \pi_{t}(i).$$
The process $(\pi_{t}(i))_{i \in {\cal S}}$ is characterized via the following system of equations
\begin{equation}\label{system}
d \pi_{t}(i) =  \sum_{j \in  {\cal S}} a_{ji} \pi_{t}(j)dt + \int_0^{+ \infty} w^\pi_{t}(i,z)(m(dt,dz) - \sum_{j \in  {\cal S}} \lambda(t,j)  F_Z(t, j, dz)\pi_{t^-}(j) dt), \quad i \in {\cal S},
\end{equation}
where
$$
w^\pi_{t}(i,z) = { d(\lambda(t,i)  F_Z(t, i, dz)\pi_{t^-}(i) ) \over 
 d( \sum_{j \in  {\cal S}} \lambda(t,j)  F_Z(t, j, dz) \pi_{t^-}(j)) }(z) - \pi_{t^-}(i),
$$
and we deduce Eq. \eqref{eqn:filter2} in Remark \ref{remark:filterapp}.\\
When $F_Z(t, i, dz)$ admits  density $f_Z(t,i,z)$, $ i \in {\cal S}$, it simplifies to
$$
w^\pi_{t}(i,z) = { \lambda(t,i)  f_Z(t, i, z)\pi_{t^-}(i)  \over 
 \sum_{j \in  {\cal S}} \lambda(t,j)  f_Z(t, j, z) \pi_{t^-}(j) } - \pi_{t^-}(i).
$$
In particular when ${\cal S}$ is a finite set, the system \eqref{system} is finite. 
This case has been considered in \cite{liangbayraktar:optreins}, with the simplification of $ \lambda(t,i)$ and $f_Z(t, j, z)$ not dependent on time.
\end{example}

\begin{example}[Markov Modulated Risk Model with known jump size distribution and unknown intensity]
\label{example:A3}
In the special case where $F_Z(t, y,dz)=F_Z(dz)$, the system \eqref{system} reduces to
\begin{equation}\label{system1}
d \pi_{t}(i) =  \sum_{j \in  {\cal S}} a_{ji} \pi_{t}(j)dt + \Big [    { \lambda(t,i)  \pi_{t^-}(i) \over 
\sum_{j \in  {\cal S}} \lambda(t,j)  \pi_{t^-}(j) }    -  \pi_{t^-}(i) \Big ] (dN_t - \sum_{j \in  {\cal S}} \lambda(t,j) \pi_{t^-}(j) dt), \quad i \in {\cal S}.
\end{equation}
Between two consecutive jump times, $t \in (T_{n-1}, T_n)$:
$$d \pi_{t}(f) = [  \sum_{j \in  {\cal S}} a_{ji} \pi_{t}(j) - \lambda(t,i)  \pi_{t}(i)   + \pi_{t}(i) \sum_{j \in  {\cal S}} \lambda(t,j) \pi_{t}(j) ] dt,$$
at a jump time $T_n$:
$$ \pi_{T_n}(i) = W_i(T_n, \pi_{T_n^-}) \doteq {  \lambda(T_n,i) \pi_{T_n^-}(i)   \over  \sum_{j \in  {\cal S}} \lambda(T_n,j) \pi_{T_n^-}(j) }.$$
This latter formula provides Eq. \eqref{eqn:filter3} in Remark \ref{remark:filterapp}.\\
In particular when ${\cal S}$ is a finite set, the system \eqref{system1}  is finite. 
\end{example}

%--------------------------------------------------------------------------------
%	BIBLIOGRAPHY
%--------------------------------------------------------------------------------
%\newpage
\bibliographystyle{apalike}
\bibliography{biblio}

\begin{thebibliography}{}

\bibitem[Brachetta and Ceci, 2019a]{BCrisks}
Brachetta, M. and Ceci, C. (2019a).
\newblock Optimal excess-of-loss reinsurance for stochastic factor risk models.
\newblock {\em Risks}, 7(2).

\bibitem[Brachetta and Ceci, 2019b]{BC:IME2019}
Brachetta, M. and Ceci, C. (2019b).
\newblock Optimal proportional reinsurance and investment for stochastic factor
  models.
\newblock {\em Insurance: Mathematics and Economics}, 87:15 -- 33.

\bibitem[Br\'emaud, 1981]{bremaud:pointproc}
Br\'emaud, P. (1981).
\newblock {\em Point Processes and Queues. Martingale dynamics}.
\newblock Springer-Verlag.

\bibitem[Bäuerle and Rieder, 2007]{bauerle2007}
Bäuerle, N. and Rieder, U. (2007).
\newblock Portfolio optimization with jumps and unobservable intensity process.
\newblock {\em Mathematical Finance}, 17(2):205--224.

\bibitem[Ceci, 2004]{ceci:2011}
Ceci, C. (2004).
\newblock Optimal investment problems with marked point stock dynamics.
\newblock In {\em Progress in Probability}, volume~63, pages 385--412.
  BirkhŁauser Verlag Basel/Switzerland.

\bibitem[Ceci, 2006]{ceci:2006}
Ceci, C. (2006).
\newblock Risk minimizing hedging for a partially observed high frequency data
  model.
\newblock {\em Stochastics}, 78(1):13--31.

\bibitem[Ceci, 2012]{ceci:2012IJTAF}
Ceci, C. (2012).
\newblock Utility maximization with intermediate consumption under restricted
  information for jump market models.
\newblock {\em International Journal of Theoretical and Applied Finance},
  15(06):1250040.

\bibitem[Ceci and Colaneri, 2012]{cc:AAP2012}
Ceci, C. and Colaneri, K. (2012).
\newblock Nonlinear filtering for jump diffusion observations.
\newblock {\em Advances in Applied Probability}, 44(3):678--701.

\bibitem[Ceci and Colaneri, 2014]{cc:AMO2014}
Ceci, C. and Colaneri, K. (2014).
\newblock The {Z}akai equation of nonlinear filtering for jump-diffusion
  observations: existence and uniqueness.
\newblock {\em Applied Mathematics and Optimization}, 69(1):47--82.

\bibitem[Ceci and Gerardi, 2006]{cg:2006}
Ceci, C. and Gerardi, A. (2006).
\newblock A model for high frequency data under partial information: a
  filtering approach.
\newblock {\em International Journal of Theoretical and Applied Finance},
  9(4):555--576.

\bibitem[Ceci and Gerardi, 2011]{ceci:2011DEF}
Ceci, C. and Gerardi, A. (2011).
\newblock Utility indifference valuation for jump risky assets.
\newblock {\em Decisions in Economics and Finance}, 34(2):85--120.

\bibitem[Confortola and Fuhrman, 2013]{Confortola2013}
Confortola, F. and Fuhrman, M. (2013).
\newblock Backward stochastic differential equations and optimal control of
  marked point processes.
\newblock {\em SIAM J. Control and Optimization}, 51:3592--3623.

\bibitem[Delong, 2013]{delong2013}
Delong, L. (2013).
\newblock {\em Backward stochastic differential equations with jumps and their
  actuarial and financial applications. BSDEs with jumps}.

\bibitem[El~Karoui et~al., 1997]{karoui:1997}
El~Karoui, N., Peng, S., and Quenez, M.~C. (1997).
\newblock Backward stochastic differential equations in finance.
\newblock {\em Mathematical Finance}, 7(1):1--71.

\bibitem[Ethier and Kurtz, 1986]{ek:1986}
Ethier, S. and Kurtz, T. (1986).
\newblock {\em Markov Processes: Characterization and Convergence}.
\newblock John Wiley \& Sons.

\bibitem[Grandell, 1991]{grandell:risk}
Grandell, J. (1991).
\newblock {\em Aspects of risk theory}.
\newblock Springer-Verlag.

\bibitem[Irgens and Paulsen, 2004]{irgens_paulsen:optcontrol}
Irgens, C. and Paulsen, J. (2004).
\newblock Optimal control of risk exposure, reinsurance and investments for
  insurance portfolios.
\newblock {\em Insurance: Mathematics and Economics}, 35:21--51.

\bibitem[Liang and Bayraktar, 2014]{liangbayraktar:optreins}
Liang, Z. and Bayraktar, E. (2014).
\newblock Optimal reinsurance and investment with unobservable claim size and
  intensity.
\newblock {\em Insurance: Mathematics and Economics}, 55:156--166.

\bibitem[Lim and Quenez, 2011]{lim2011}
Lim, T. and Quenez, M.-C. (2011).
\newblock Exponential utility maximization in an incomplete market with
  defaults.
\newblock {\em Electron. J. Probab.}, 16:1434--1464.

\bibitem[Liu and Ma, 2009]{liuma:optreins}
Liu, B. and Ma, J. (2009).
\newblock Optimal reinsurance/investment problems for general insurance models.
\newblock {\em The Annals of Applied Probability}, 19:1495–1528.

\bibitem[Rolski et~al., 1999]{rolski:insurancefin}
Rolski, T., Schmidli, H., V., S., and Teugels, J. (1999).
\newblock {\em Stochastic processes for insurance and finance}.
\newblock Wiley.

\bibitem[Schmidli, 2008]{schmidli:control}
Schmidli, H. (2008).
\newblock {\em Stochastic Control in Insurance}.
\newblock Springer-Verlag.

\bibitem[Schmidli, 2018]{schmidli:2018risk}
Schmidli, H. (2018).
\newblock {\em Risk Theory}.
\newblock Springer Actuarial. Springer International Publishing.

\bibitem[Young, 2006]{young:premium_princ}
Young, V.~R. (2006).
\newblock Premium principles.
\newblock {\em Encyclopedia of Actuarial Science}, (3).

\end{thebibliography}

\end{document}